\documentclass[sort&compress]{elsarticle}
\usepackage{amsmath,mathtools,amssymb}
\usepackage{graphicx}   
\usepackage{epstopdf}
\usepackage{subfigure}
\usepackage{color}
\usepackage{verbatim} 
\usepackage{booktabs}
\usepackage{bm}   
\usepackage{dutchcal}
\usepackage[left=3cm,right=3cm,top=2.5cm,bottom=2.5cm]{geometry}  
\usepackage{lineno,hyperref}
\usepackage{tablists} 
\usepackage{makecell} 
\usepackage{tikz}
\usepackage{hyperref}
\usepackage{pstricks}
\usepackage[titletoc]{appendix}
\usepackage{pstricks-add}
\usepackage{float}
\usepackage{mathrsfs} 
\usepackage{amsthm}
\hypersetup{
	colorlinks=ture,
	linkcolor=blue,
	filecolor=gray,
	urlcolor=blue,
	citecolor=blue}
\numberwithin{equation}{section}
\newtheorem{definition}{Definition}[section]

\newtheorem{theorem}{Theorem}[section]
\newtheorem{lemma}{Lemma}[section]	
\newtheorem{proposition}{Proposition}[section]
\allowdisplaybreaks[4]

\begin{document}
	\begin{frontmatter}
		\title{Higher Chern-Simons based on (2-)crossed modules} 
		\author{Danhua Song\corref{cor1}}	
		\ead{danhua_song@163.com}
		\author{Mengyao Wu}
		\ead{2210502108@cnu.edu.cn}
		\author{Ke Wu}
		\ead{wuke@cnu.edu.cn}
		\author{Jie Yang}
	\ead{yangjie@cnu.edu.cn}
			\cortext[cor1]{Corresponding author.}
		\address{School of Mathematical Sciences, Capital Normal University, Beijing 100048, China}
		\date{}
		
		\begin{abstract} 
			We present higher Chern-Simons theories  based on (2-)crossed modules.
			We start from the generalized differential forms in Generalized Differential Calculus and define the corresponding generalized connections which consist of higher connections. Then we establish the generalized Chern-Simons forms to get the higher Chern-Simons actions. Finally, we develop the higher second Chern forms and  Chern-Weil theorems.
			
		\end{abstract}
		
	\end{frontmatter}
	%
	

	\clearpage
	\tableofcontents

		\section{Introduction}
	The purpose of this paper is to construct Higher Chern-Simons (HCS) gauge theories based on (2-)crossed modules. 
	The motivation for this work comes from literature  \cite{HSUSJS, DFUSJS, DFCLRUS, PRCS} and Soncini and Zucchini's papers \cite{ESRZ, R.Z, R.Z2, R.Z3, R.Z4} in which they  formulate a 4-dimensional  semistrict higher gauge theoretic Chern-Simons (CS) theory. But we only consider the case of the strict higher gauge theories. Besides,
	our treatments and techniques adopted in constructions are significantly different from that in literature.
	Within the framework of higher gauge theory \cite{Baez.2010, FH},
	we construct 2-Chern-Simons (2CS) and 3-Chern-Simons (3CS) gauge theories by applying
	the Generalized Differential Calculus (GDC) \cite{Robinson1}. 

	The higher gauge theory has attracted considerable attentions in many branches of physics, such as the $D$-branes and $M$-branes in the string
	theory \cite{JP, KBMBJHS, CVJ}, the spin foam models in the loop quantum gravity \cite{Baez, CR}, the 6-dimensional superconformal field theory \cite{Saemann:2013pca} and so on. In the past two decades, the 2-gauge theory \cite{ACJ, Baez2005HigherGT, Bar, JBUS} and the 3-gauge theory \cite{Faria_Martins_2011, doi:10.1063/1.4870640, Saemann:2013pca, DFHSUS} have been studied deeply. Based on  these higher gauge theories, the electromagnetic theory has been generalized to the $p$-form electromagnetic theory \cite{HP, MHCT, MKPR}. Likewise, the Yang-Mills (YM) theory has been generalized to the 2-form Yang-Mills (2YM) theory \cite{2002hep.th....6130B} and the 3-form Yang-Mills (3YM) theory \cite{Song}, respectively. Furthermore, the topological BF theory has been generalized to the topological 2BF \cite{Martins:2010ry, FGHEMP, AMMAO, AMMAOM, AMMAOMV, AMMAO2} and 3BF theories \cite{TRMV1, Radenkovic:2019qme, AMM}.
	The development of the higher YM and BF theories motivates us to consider the higher counterparts of the CS theory by following the idea on the higher gauge theory.

	As an ordinary gauge theory, the CS theory has a wide variety of applications, for instance pure mathematics, string theory, condensed matter physics \cite{CS, NM, BJMJ}, and so on. Since the CS theory does not depend on the metric of the underlying spacetime manifold, it is known as a sort of topological field theory \cite{GTH, DBMBMRGT}. The CS action is constructed by the connection 1-form on a principle bundle, which is called a ``gauge field'', and the equation of motion implies the flatness of the connection, i.e. the corresponding curvature 2-form vanishes. Without loss of generality, similar arguments shall be applied to the 2CS and 3CS theories. In carrying out the developments of these HCS theories, the main difficulty is the construction of the 2CS and 3CS actions. The main objects are the 2-connections and 3-connections, which consist of two and three ordinary differential forms valued in higher algebras respectively.
	We find out that the choice of the generalized differential forms in the GDC seems to be the best choice to our theories.

	The GDC has been discussed in a number of literature and they are employed in many different geometrical and physical contexts. For example, in 2002, Guo et al. introduced the generalized topological field theories \cite{GHY}, and found that a direct relation between CS and BF theories can be presented by using the GDC. Besides, they established the generalized Chern-Weil homomorphism for generalized curvature invariant polynomials, and showed that the BF gauge theory can be obtained from the generalized second Chern class. In 2003, they presented a general approach to construct a class of generalized topological field theories with constraints  and found that the ordinary BF formulations of general relativity, Yang-Mills theories, and $\mathscr{N}=1, 2$ chiral super-gravity can  be reformulated as generalized topological fields with constraints \cite{ghy}. 	Recently, the GDC has been developed by Robinson in a number of works  \cite{Robbinson0, Robinson2, Robbinson4, Robbinson5, Robinson3}.
	In the mean time, analogous ideas have been explored in the dual context, and generalized vector fields have been introduced and studied in  \cite{Robinson6, gv}. 
	
	In the GDC, the algebra and calculus of ordinary differential forms are extended to an algebra and calculus of different type $N$ generalized differential forms.  We are concerned in this paper with the type $N=1$ and $N=2$, which consist of two and three ordinary differential forms, respectively. Based on these facts, it is possible to redefine the 2-connection as a type $N=1$ generalized differential form, and the 3-connection as a type $N=2$ generalized differential form.
	The advantage of using the GDC lies in the fact that we can rewrite the higher connection as a kind of generalized connection.           
	
	Analogous to higher BF and YM theories, the CS theory will be generalized to the higher counterpart by using the idea of a categorical ladder. As indicated in Table \ref{table 1}, the underlying algebraic structure  is promoted from a ordinary group to higher groups.
	\begin{center}
		\begin{table}[H]
			\caption{The mathematical structures of (higher) CS theories}\label{table 1}
			\centering
			\setlength{\tabcolsep}{4mm}{
				\begin{tabular}{cccc}
					\specialrule{1pt}{5pt}{1pt} \specialrule{0.5pt}{1pt}{3pt}
					Gauge theory & Categorical structure& Algebraic structure& Linear structure   \\[1mm] 
					\specialrule{0.5pt}{2pt}{6pt} 
					CS & Lie group & Lie group & Lie algebra  \\[1.5mm]
					\specialrule{0.5pt}{1pt}{7pt}  
					2CS & Lie 2-group & Lie crossed module &
					\begin{tabular}[c]{@{}l@{}}\ \ \ \ differential  \\\ crossed module \end{tabular} \\[1.5mm]
					\specialrule{0.5pt}{1pt}{6pt} 
					3CS & Lie 3-group& Lie 2-crossed module &  \begin{tabular}[c]{@{}l@{}} \ \ \ \ \ differential \\\
						2-crossed module \end{tabular}   \\[2.5mm]  
					\specialrule{0.5pt}{2pt}{1pt} \specialrule{1pt}{1pt}{0pt} 
			\end{tabular} } 
		\end{table}
	\end{center}
	It follows that the Lie algebra-valued connection form shall be replaced by the differential (2-)crossed module-valued connection form. 
	In order to construct the 2CS and 3CS actions, we need to consider the invariant bilinear form on the differential (2-)crossed module. We follow the definition of the invariant form on the differential crossed module in \cite{ESRZ, R.Z}. Motivated by this construction, we define an invariant form on the differential 2-crossed module.
	Furthermore, we establish generalized bilinear forms based on the above expansions and establish the generalized CS forms to give the HCS forms and actions.

	According to the ordinary CS gauge theories, three questions naturally arise:
	\begin{enumerate}
		\item Whether higher second Chern forms exist?
		\item Whether the higher second Chern forms are gauge invariant under higher gauge transformations?
		\item Whether the higher second Chern forms satisfy higher Chern-Weil theorems?
	\end{enumerate}
	For the first question, we establish two new second Chern forms called respectively second 2-Chern form and second 3-Chern form, which are constructed under the categorical generalizations. For the second question, we prove that the second 2-Chern form is 2-gauge invariant under the 2-gauge transformation, and the second 3-Chern form is 3-gauge invariant under the 3-gauge transformation.
	For the third question,
	we show that the two new second Chern forms are closed and have topologically invariants, which give 2- and 3-Chern-Weil theorems separately. 
	We list  our results (2CS and 3CS) and the known results in Table \ref{table 2}. We note that the higher CS theories are nice generalizations of the CS theory. 
	\begin{center}
		\begin{table}[H]
			\caption{The (higher) CS gauge theories}\label{table 2}
			\centering
			\setlength{\tabcolsep}{1mm}{
				\begin{tabular}{cccc}
					\specialrule{1pt}{5pt}{1pt} \specialrule{0.5pt}{1pt}{3pt}
					(Higher) Gauge theory & CS   &2CS   &3CS\\[1mm] 
					\specialrule{0.5pt}{2pt}{6pt} 
					Gauge field&$A$&$(A, B)$ & $(A, B, C)$ \\[1.5mm]
					\specialrule{0.5pt}{2pt}{6pt}  
					Lagrangian	 & $\langle A, F- \frac{1}{3}A\wedge A\rangle_{\mathcal{g}}$ &$ \langle 2 \Omega_1+ \alpha(B), B \rangle_{\mathcal{g}, \mathcal{h}}$ & $ \langle 2\Omega_1 + \alpha(B), C\rangle_{\mathcal{g}, \mathcal{l}} +\langle B, \Omega_2 \rangle_{\mathcal{h}}$ \\[1.5mm]
					\specialrule{0.5pt}{2pt}{6pt} 
					Second Chern form 	 &$\langle F, F \rangle_{\mathcal{g}}$&  $2\langle \Omega_1, \Omega_2 \rangle_{\mathcal{g}, \mathcal{h}}$  &$2 \langle \Omega_1, \Omega_3 \rangle_{\mathcal{g}, \mathcal{l}} + \langle \Omega_2, \Omega_2 \rangle_{\mathcal{h}}$ \\[1.5mm]  
					\specialrule{0.5pt}{2pt}{6pt} 
					Equations of motion 	 &$F=0$&  $\Omega_1= \Omega_2=0$  &$\Omega_1= \Omega_2=\Omega_3=0$ \\[1.5mm]  
					\specialrule{0.5pt}{2pt}{1pt} \specialrule{1pt}{1pt}{0pt} 
			\end{tabular} } 
		\end{table}
	\end{center}
	
	The layout of the paper is summarized in the following outline:
	\begin{enumerate}
		\item In section \ref{sub2}, we review the relevant topics of the higher gauge theories. The notations of Lie algebra-valued differential forms and some conventions are presented in subsection \ref{sub21}. The 2-connection and its 2-gauge transformations are given in \ref{2connection}. Similarly, the 3-connection and its 3-gauge transformations are listed in subsection \ref{3connection}.
		\item In section \ref{sub3}, we recall the balanced differential crossed modules and the related invariant forms, which will play major roles in the 2CS theory. Motivated by these arguments, we establish the balanced differential 2-crossed modules and the relational invariant forms, which will play similar roles in the 3CS theory.
		\item In section \ref{sub4}, we recall the type $N$ generalized differential forms in the GDC and present the type $N=1$ and $N=2$ generalized $p$-forms. Then, we develop the two types generalized forms valued in differential crossed modules and 2-crossed modules.
		\item In section \ref{sub5}, we rewrite the 2-connections as the type $N=1$ generalized connections, and the 3-connections as the type $N=2$ generalized connections based on the above GDC.
		\item In section \ref{sub6}, 
		we present the two main results of this paper.  Firstly, in subsection \ref{sub61} we recollect the ordinary second Chern form and Chern-Weil theorem and the CS form. 
		Secondly, in subsection \ref{sub62}, we introduce the 2CS form based on the type $N=1$ generalized connections, and generalize the second Chern form to the second 2-Chern form satisfying the 2-Chern-Weil theorem.
		Finally, the similar construction apply to the type $N=2$ generalized connections, and we get the 3CS form in subsection \ref{sub63}. In addition, we also build the second 3-Chern form and show that it fulfills the 3-Chern-Weil theorem.
	\end{enumerate}

	%
	%
	%
	%
	%
	%
	%
	%
	%
	%
	%
	%
	%
	%
	%
	%
	%
	%

	\section{Higher connections and gauge transformations}\label{sub2}
	In this section, we shall review the higher connections and associated gauge transformations. Since we focus on the construction of HCS gauge theories, we consider these higher gauge fields (or higher connections) as basic fields of these theories.


	\subsection{Higher algebra valued differential forms}\label{sub21}
	
	In the ordinary CS gauge theory with gauge group $G$, the gauge fields are  $\mathcal{g}$-valued differential forms. Similarly, the higher gauge fields can be described by higher algebra valued differential forms in the HCS gauge theories. 
	Both Lie crossed modules and Lie 2-crossed modules are the algebraic structures of the HCS gauge theories elaborated in this paper rests.  We will not go into all the details about these concepts,  however a fuller introduction is given in Appendix \ref{2cm} for serving also the purpose of setting the notations used throughout this paper. 
	Furthermore, see refs. \cite{Faria_Martins_2011, 2002hep.th....6130B, Beaz, Crans, Brown, Kamps20022groupoidEI, Mutlu1998FREENESSCF}  for an exhaustive exposition of this subject. As for prerequisites, it is enough to be familiar with the following definitions.
	
	A Lie crossed module $(H, G; \alpha, \vartriangleright)$ consists of two Lie groups $H$ and $G$ together with a Lie group action $\vartriangleright: G \times H \longrightarrow H$ of $G$ on $H$ by automorphisms, and  an equivariant target map $\alpha : H \longrightarrow G$ satisfying certain identities. 
	The Lie crossed module has an infinitesimal version, $(\mathcal{h}, \mathcal{g}; \alpha, \vartriangleright)$, called the differential crossed module. In order not to introduce additional notations, we use the same letters $\alpha$ and $\vartriangleright$ for counterparts in the infinitesimal version of the crossed module.  The convention also applies to the Lie 2-crossed modules.
	
	A Lie 2-crossed module $(L, H, G; \beta, \alpha, \vartriangleright, \{, \})$ is a set of three groups $L$, $H$, $G$ with
	a complex of Lie groups $L\stackrel{\beta}{\longrightarrow}H \stackrel{\alpha}{\longrightarrow}G$ and actions $\vartriangleright$ of $G$ on $G$ by conjugation, and on $H$, $L$ by automorphisms, as well as a $G$-equivariant smooth function $\{-, -\} : H \times H \longrightarrow L$, called the Peiffer lifting. 	If $(H, G; \alpha, \vartriangleright)$ is also a crossed module in $(L, H, G; \beta, \alpha, \vartriangleright, \{, \})$, we call this 2-crossed module fine. In this paper, we focus on the fine 2-crossed module.
	Similarly, the Lie 2-crossed module has an  infinitesimal version, $(\mathcal{l}, \mathcal{h}, \mathcal{g}; \beta, \alpha, \vartriangleright, \{, \})$, called the differential 2-crossed module.

	Then, we consider the differential forms valued in these associated algebras as follows.
	In this paper, we will mostly follow the notations and definitions of \cite{doi:10.1063/1.4870640, SDH}. 	We denote by $\Lambda^k(M, \mathcal{g})$ the vector space of $\mathcal{g}$-valued differential $k$-forms on a manifold $M$ over $C^{\infty}(M)$.
	For $A\in \Lambda^k (M, \mathcal{g})$, have $A=\sum\limits_{a}A^a X_a $ with a scalar differential $k$-form $A^a$ and an element $X_a$ in $\mathcal{g}$. We define
	\begin{align}
		dA := \sum\limits_{a} dA^a X_a.
	\end{align}
	Here, we assume $\mathcal{g}$ to be a matrix Lie algebra. Then, we have $\left[X, X'\right]=XX'-X'X$ for each $X, X'\in \mathcal{g}$.
	For $A_1=\sum\limits_{a}A^a_1 X_a \in \Lambda^{k_1} (M, \mathcal{g}) $, $A_2=\sum\limits_{b}A^b_2 X_b \in \Lambda^{k_2} (M, \mathcal{g}) $, 
	define
	\begin{align*}
		A_1 \wedge A_2 :=\sum\limits_{a,b}A^a_1 \wedge A^b_2 X_a X_b,\ \ \ \ 
		A_1 \wedge^{\left[, \right]} A_2 :=\sum\limits_{a,b}A^a_1 \wedge A^b_2 \left[X_a, X_b\right], 
	\end{align*}
	then there is an identity
	$$A_1\wedge^{[, ]} A_2=A_1\wedge A_2-(-1)^{k_1 k_2}A_2\wedge A_1.$$
	The convention also applies to $\mathcal{h}$ and $\mathcal{l}$.
	
	Besides, for $B=\sum\limits_{a} B^a Y_a \in \Lambda^{t_1}(M,\mathcal{h})$, $B'=\sum\limits_{b} B'^b Y_b \in \Lambda^{t_2}(M,\mathcal{h})$ with $Y_a$, $Y_b \in \mathcal{h}$, define
	\begin{align*}
		A\wedge^{\vartriangleright }B :=\sum\limits_{a,b} A^a \wedge B^b X_a \vartriangleright Y_b,\ \ 	\alpha (B) := \sum\limits_{a}B^a \alpha(Y_a),\\
		B\wedge^{\left\{, \right\}} B':= \sum\limits_{a,b} B^a \wedge B'^b \left\{Y_a,Y_b\right\}.
	\end{align*}
	For $C=\sum\limits_{a}C^a Z_a \in \Lambda^{q_1} (M, \mathcal{l})$, $C'=\sum\limits_{b}C'^b Z_b \in \Lambda^{q_2} (M, \mathcal{l})$, define
	\begin{align*}
		A\wedge^\vartriangleright C :=\sum\limits_{a,b} A^a \wedge C^b X_a \vartriangleright Z_b, \ \  \beta(C):=\sum\limits_{a}C^a \beta(Z_a),\\
		B\wedge^{\vartriangleright'}C:=\sum\limits_{a,b}B^a \wedge C^b Y_a\vartriangleright' Z_b\
	\end{align*}
	where $Y_a \vartriangleright' Z_b = -\left\{\beta(Z_b), Y_a\right\}$ by using \eqref{YZZ}.
	
	For more properties for the Lie algebra valued differential forms corresponding to the identities of the differential (2-)crossed module, we refer the reader to \cite{SDH}.
	
	\subsection{2-connections and 2-gauge transformations}\label{2connection}
	
	Given a crossed module $O=(H, G; \alpha, \vartriangleright)$, let $\mathcal{O}=(\mathcal{h}, \mathcal{g}; \alpha, \vartriangleright)$ be the associated differential crossed module. 	The basic gauge fields of 2-gauge theory are 2-connections valued in the differential crossed module.
	On the manifold $M$, a 2-connection $(A, B)$ is given by a $\mathcal{g}$-valued 1-form $A \in \Lambda^1(M, \mathcal{g})$ and an $\mathcal{h}$-valued 2-form $B \in \Lambda^2(M, \mathcal{h})$. 
	The curvature $(\Omega_1, \Omega_2)$ consists of a $\mathcal{g}$-valued fake curvature 2-form and an $\mathcal{h}$-valued 2-curvature 3-form:
	\begin{align}\label{1}
		\Omega_1= dA + A \wedge A - \alpha(B),\ \ \ \
		\Omega_2= dB + A \wedge^{\vartriangleright}B, 
	\end{align}
	and call $(A, B)$ fake-flat, if $\Omega_1=0$, and flat, if it is fake-flat and $\Omega_2=0$. 
	
	On the other hand, the curvature automatically satisfies the 2-Bianchi Identities:
	\begin{align}\label{2BI}
		d \Omega_1 + A \wedge^{[, ]}\Omega_1 + \alpha(\Omega_2)&=0,\nonumber\\
		d \Omega_2 + A \wedge^{\vartriangleright} \Omega_2 - \Omega_1 \wedge^{\vartriangleright} B &=0.
	\end{align}
	
	Moreover, there are two kinds of 2-gauge transformations from a 2-connection $(A, B)$ to another one $(A', B')$:
	\begin{itemize}
		\item the first kind of 2-gauge transformation with $g \in \Lambda^0(M, G)$,
		\begin{align}
			A'=g^{-1}A g + g^{-1}dg, \ \ \ \ 
			B'= g^{-1}\vartriangleright B; \label{11}
		\end{align} 
		\item the second kind of 2-gauge transformation with $ \phi \in \Lambda^1(M, \mathcal{h})$,
		\begin{align}
			A'= A + \alpha(\phi), \ \ \ \ 
			B'= B + d \phi + A' \wedge^{\vartriangleright}\phi - \phi \wedge \phi. 
		\end{align}
	\end{itemize}
	If we write the second kind of 2-gauge transformation as
	\begin{align}
		A''= A' + \alpha(\phi), \ \ \ \ 
		B'' = B' + d \phi + A'' \wedge^{\vartriangleright}\phi - \phi \wedge \phi, \label{22}
	\end{align}
	then the composition of \eqref{11} and \eqref{22} gives a general 2-gauge transformation:
	\begin{align}\label{33}
		&A''=  g^{-1}A g + g^{-1}dg + \alpha(\phi),\nonumber\\
		&B''= g^{-1}\vartriangleright B + d \phi + A'' \wedge^{\vartriangleright}\phi - \phi \wedge \phi. 
	\end{align}
	Under the general 2-gauge transformation, the associated curvature transforms as follows:
	\begin{align}\label{2gtc}
		&\Omega''_1 = g^{-1} \Omega_1 g,\nonumber\\
		& \Omega''_2 = g^{-1} \vartriangleright \Omega_2 + \Omega''_1 \wedge^{\vartriangleright} \phi.
	\end{align}
	There is a basic result that the 2-connection is fake flat on the trivial principal $O$-2-bundle over $M$. The 2-curvature 3-form $\Omega_2$ is thus covariant under the general 2-gauge transformation. 
	
	\begin{theorem}
		Under the general 2-gauge transformation \eqref{33}, the 2-connection $(A, B)$ has the following properties:
		\begin{itemize}
			\item[(a).]  The fake-flatness is 2-gauge invariant;
			\item[(b).] The flatness is 2-gauge invariant;
			\item[(c).] The 2-Bianchi Identities are 2-gauge invariant.
		\end{itemize}
	\end{theorem}
	\begin{proof}
		According to the definitions of (fake-)flatness and the transformations \eqref{2gtc}, it is straightforward to get (a) and (b).  Besides, the proof of (c) is not difficult but is too long to give here. We will give a rigorous proof in Appendix \ref{2gt}.
	\end{proof}
	
	\subsection{3-connections and 3-gauge transformations}\label{3connection}
	
	Given a 2-crossed module $W=(L, H, G; \beta, \alpha, \vartriangleright, \{, \})$, let $\mathcal{W}=(\mathcal{l}, \mathcal{h}, \mathcal{g}; \beta, \alpha, \vartriangleright, \{, \})$ be the associated differential 2-crossed module.
	The basic gauge fields of 3-gauge theory are 3-connections valued in the differential 2-crossed module.
	On a manifold $M$, a 3-connection $(A, B,C)$ is given by a $\mathcal{g}$-valued 1-form $A \in \Lambda^1(M, \mathcal{g})$, an $\mathcal{h}$-valued 2-form $B \in \Lambda^2(M, \mathcal{h})$, and an $\mathcal{l}$-valued 3-form $C \in \Lambda^3(M, \mathcal{l})$. The curvature $(\Omega_1, \Omega_2, \Omega_3)$ consists of a $\mathcal{g}$-valued fake curvature 2-form, an $\mathcal{h}$-valued fake 2-curvature 3-form and an $\mathcal{l}$-valued 3-curvature 4-form:
	\begin{align}\label{2}
		&\Omega_1=dA + A\wedge A - \alpha(B), \ \ 
		&\Omega_2=dB + A \wedge^{\vartriangleright}B - \beta(C),\ \ \ \
		&\Omega_3= dC + A \wedge^{\vartriangleright}C + B \wedge^{\{, \}}B,
	\end{align}
	and call $(A, B, C)$ fake 1-flat, if $\Omega_1=0$, and fake flat, if it is fake 1-flat and $\Omega_2=0$, and flat, if it is fake flat and $\Omega_3=0$.
	
	Starting from the definition of the curvature form, we can also calculate its exterior
	derivative to obtain the 3-Bianchi Identities:
	\begin{align}\label{3bi}
		d \Omega_1 + A \wedge^{[, ]} \Omega_1 + \alpha(\Omega_2)&=0,\nonumber \\
		d \Omega_2 + A \wedge^{\vartriangleright}\Omega_2 - \Omega_1 \wedge^{\vartriangleright}B + \beta(\Omega_3)&=0, \nonumber \\
		d \Omega_3 + A \wedge^{\vartriangleright}\Omega_3 - \Omega_1\wedge^{\vartriangleright}C - B\wedge^{\{, \}}\Omega_2 - \Omega_2 \wedge^{\{, \}}B &=0.
	\end{align}
	For a rigorous derivation of these identities, the reader is referred to \cite{Song}.
	
	Without loss of generality, we consider three kinds of 3-gauge transformations from a 3-connection $(A, B, C)$ to another one $(A', B', C')$:
	\begin{itemize}
		\item the first kind of 3-gauge transformation with $g \in \Lambda^0(M, G)$,
		\begin{align}\label{31}
			A'= g^{-1}A g + g^{-1}dg, \ \ \ \ 
			B'= g^{-1}\vartriangleright B,\ \ \ \ 
			C'= g^{-1}\vartriangleright C;
		\end{align}
		\item the second kind of 3-gauge transformation with $\phi \in \Lambda^1(M, \mathcal{h})$,
		\begin{align}
			A'= A + \alpha(\phi), \ \ \ \ 
			B' = B + d \phi + A' \wedge^{\vartriangleright}\phi - \phi \wedge \phi,\ \ \ \ 
			C'= C - B' \wedge^{\{, \}} \phi - \phi \wedge^{\{, \}} B;
		\end{align}
		\item  the third kind of 3-gauge transformation with $\psi \in \Lambda^2(M, \mathcal{l})$,
		\begin{align}
			A'= A, \ \ \ \ 
			B'= B - \beta(\psi), \ \ \ \ 
			C'= C- d \psi - A' \wedge^{\vartriangleright}\psi.
		\end{align}
	\end{itemize}
	
	If we write the second kind of 3-gauge transformation as
	\begin{align}\label{32}
		A''= A'+ \alpha(\phi), \ \ \ \ 
		B''= B' + d \phi + A'' \wedge^{\vartriangleright}\phi - \phi \wedge \phi, \ \ \ \ 
		C''= C' - B'' \wedge^{\{, \}} \phi - \phi \wedge^{\{, \}} B',  
	\end{align}
	and the third kind of 3-gauge transformation as
	\begin{align}\label{34}
		\overline{A}=A'', \ \ \ \ 
		\overline{B}= B'' - \beta(\psi), \ \ \ \ 
		\overline{C}= C''- d \psi - \overline{A} \wedge^{\vartriangleright}\psi,
	\end{align}
	then the composition of \eqref{31}, \eqref{32} and \eqref{34} gives a general 3-gauge transformation:
	\begin{align}\label{g3gt}
		\overline{A}&= g^{-1}A g +g^{-1}dg + \alpha(\phi), \nonumber\\
		\overline{B}&= g^{-1}\vartriangleright B + d \phi + \overline{A}\wedge^{\vartriangleright}\phi - \phi \wedge \phi -\beta(\psi), \nonumber\\
		\overline{C}&=  g^{-1}\vartriangleright C - \overline{B}\wedge^{\{, \}} \phi + \phi \wedge^{\vartriangleright'}\psi - \phi \wedge^{\{, \}}(g^{-1}\vartriangleright B)- d \psi - \overline{A} \wedge^{\vartriangleright} \psi.
	\end{align}
	
	Under the general 3-gauge transformation, the 3-curvature $(\Omega_1, \Omega_2, \Omega_3)$ transforms as follows:
	\begin{align}\label{3gtc}
		\overline{\Omega}_1& = g^{-1}\Omega_1 g, \nonumber\\
		\overline{\Omega}_2& = g^{-1}\vartriangleright \Omega_2 + \overline{\Omega}_1\wedge^{\vartriangleright}\phi,\nonumber\\
		\overline{\Omega}_3& = g^{-1}\vartriangleright \Omega_3 - \overline{\Omega}_2 \wedge^{\{, \}} \phi + \phi \wedge^{\{, \}} (g^{-1}\vartriangleright \Omega_2) - \overline{\Omega}_1\wedge^{\vartriangleright} \psi.
	\end{align}
	
	By the same token with the 2-connections, there is a similar result that the 3-connections is fake flat, i.e. $\Omega_1= \Omega_2=0$, on the trivial principal $W$-3-bundle over $M$. The 3-curvature 4-form $\Omega_3$ is then covariant under the general 3-gauge transformation.
	
	\begin{theorem}
		Under the general 3-gauge transformation \eqref{g3gt}, the 3-connection $(A, B,C)$ has the following properties:
		\begin{itemize}
			\item[(a).]  The fake 1-flatness is 3-gauge invariant;
			\item[(b).]  The fake-flatness is 3-gauge invariant;
			\item[(c).] The flatness is 3-gauge invariant;
			\item[(d).] The 3-Bianchi Identities are 3-gauge invariant.
		\end{itemize}
	\end{theorem}
	\begin{proof}
		\begin{itemize}
			\item[(a).] If $(A,B, C)$ is fake 1-fat, $\Omega_1=0$. Under the general 3-gauge transformation \eqref{g3gt},  $\overline{\Omega}_1= g^{-1}\Omega_1 g = 0$. Then $(\overline{A}, \overline{B}, \overline{C})$ is also fake 1-flat, i.e. the flat 1-flatness is 3-gauge invariant.
			\item[(b).] If $(A,B, C)$ is fake fat, $\Omega_1= \Omega_2= 0$. Under the general 3-gauge transformation \eqref{g3gt},  $\overline{\Omega}_1= g^{-1}\Omega_1 g = 0$, and $\overline{\Omega}_2= g^{-1}\vartriangleright \Omega_2 + \overline{\Omega}_1\wedge^{\vartriangleright}\phi=0$. Then $(\overline{A}, \overline{B}, \overline{C})$ is also fake flat, i.e. the flat flatness is 3-gauge invariant. 
			\item[(c).] If $(A,B, C)$ is fat, $\Omega_1= \Omega_2=\Omega_3= 0$. Under the general 3-gauge transformation \eqref{g3gt},  $\overline{\Omega}_1= g^{-1}\Omega_1 g = 0$, $\overline{\Omega}_2= g^{-1}\vartriangleright \Omega_2 + \overline{\Omega}_1\wedge^{\vartriangleright}\phi=0$, and $\overline{\Omega}_3= g^{-1}\vartriangleright \Omega_3 - \overline{\Omega}_2 \wedge^{\{, \}} \phi + \phi \wedge^{\{, \}} (g^{-1}\vartriangleright \Omega_2) - \overline{\Omega}_1\wedge^{\vartriangleright} \psi=0$. Then $(\overline{A}, \overline{B}, \overline{C})$ is also flat, i.e. the flatness is 3-gauge invariant. 
			\item[(d).] The proof of (d) is not difficult but is too long to give here. We will give a rigorous proof in Appendix \ref{2gt}.
			
		\end{itemize}
	\end{proof}

	\section{Balanced differential (2-)crossed modules with invariant forms}\label{sub3}
	
	Similar to the role of invariant traces in ordinary CS gauge theory, the invariant form is also an essential ingredient of the construction of the HCS action.
	Thus, the first key issue for developing the HCS theories is defining the invariant bilinear form on the relevant algebras. 
	
	\paragraph{\textbf{Balanced differential crossed modules}}
	Balanced differential crossed modules play a major role in the construction of the 2CS theory. The notion of balance has arisen in  \cite{ESRZ}, which has non counterpart in ordinary Lie algebra theory.
	
	\begin{definition}[Balanced differential crossed modules]
		A differential crossed module $(\mathcal{h}, \mathcal{g}; \alpha, \vartriangleright)$ is said balanced if dim $\mathcal{h}$= dim $\mathcal{g}$.
	\end{definition}
	
	\begin{proposition}
		For any non balanced differential crossed module $(\mathcal{h}, \mathcal{g}; \alpha, \vartriangleright)$, there exists a balanced differential crossed module $(\tilde{\mathcal{h}}, \tilde{\mathcal{g}}; \tilde{\alpha}, \tilde{\vartriangleright})$ minimally extending $(\mathcal{h}, \mathcal{g}; \alpha, \vartriangleright)$.
	\end{proposition}
	\begin{proof}
		Firstly,  let dim $\mathcal{g} < $ dim $\mathcal{h}$. Then, there is a balanced
		differential crossed module $(\tilde{\mathcal{h}}, \tilde{\mathcal{g}}; \tilde{\alpha}, \tilde{\vartriangleright})$ where
		\begin{itemize}
			\item[1.] $\tilde{\mathcal{h}}=\mathcal{h}$;
			\item[2.] $\tilde{\mathcal{g}}= \mathcal{g} \oplus \mathcal{w}$, where $\mathcal{w}$ is a vector space such that dim $\mathcal{w}=$ dim $\mathcal{h}$ $-$  dim $\mathcal{g}$, and the Lie bracket $[-, -]$ in the Lie algebra $\tilde{\mathcal{g}}$ is given by
			\begin{align}
				[X_1 \oplus W_1, X_2 \oplus W_2] = [ X_1, X_2] \oplus 0;
			\end{align}
			\item[3.] the map $\tilde{\alpha}: \tilde{\mathcal{h}} \longrightarrow \tilde{\mathcal{g}}$ is given by
			\begin{align}
				\tilde{\alpha}(Y)= \alpha(Y) \oplus 0;
			\end{align}
			\item[4.] the action $\tilde{\vartriangleright}: \tilde{\mathcal{g}} \times \tilde{\mathcal{h}} \longrightarrow \tilde{\mathcal{h}}$ is given by 
			\begin{align}
				(X \oplus W)\tilde{\vartriangleright} Y = X \vartriangleright Y,
			\end{align}
		\end{itemize}
		for any $X, X_1, X_2 \in \mathcal{g}, Y \in \mathcal{h}$ and $W, W_1, W_2 \in \mathcal{w}$.
		
		Secondly, let dim $\mathcal{g} > $ dim $\mathcal{h}$. Then, there is a balanced
		differential crossed module $(\tilde{\mathcal{h}}, \tilde{\mathcal{g}}; \tilde{\alpha}, \tilde{\vartriangleright})$ where 
		\begin{itemize}
			\item[1.] $\tilde{\mathcal{g}}=\mathcal{g}$;
			\item[2.] $\tilde{\mathcal{h}}= \mathcal{h} \oplus \mathcal{w'}$, where $\mathcal{w'}$ is a vector space such that dim $\mathcal{w'}=$ dim $\mathcal{g}$ $-$  dim $\mathcal{h}$, and the Lie bracket $[-, -]$ in the Lie algebra $\tilde{\mathcal{h}}$ is given by
			\begin{align}
				[Y_1 \oplus W'_1, Y_2 \oplus W'_2] = [ Y_1, Y_2] \oplus 0;
			\end{align}
			\item[3.] the map $\tilde{\alpha}: \tilde{\mathcal{h}} \longrightarrow \tilde{\mathcal{g}}$ is given by
			\begin{align}
				\tilde{\alpha}(Y\oplus W')= \alpha(Y);
			\end{align}
			\item[4.] the action $\tilde{\vartriangleright}: \tilde{\mathcal{g}} \times \tilde{\mathcal{h}} \longrightarrow \tilde{\mathcal{h}}$ is given by 
			\begin{align}
				X \tilde{\vartriangleright} (Y \oplus W') = (X \vartriangleright Y)\oplus 0,
			\end{align}
		\end{itemize}
		for any $X \in \mathcal{g}, Y_1, Y_2 \in \mathcal{h}$ and $ W', W'_1, W'_2 \in \mathcal{w'}$.
		
		Further, the differential crossed module $(\tilde{\mathcal{h}}, \tilde{\mathcal{g}}; \tilde{\alpha}, \tilde{\vartriangleright})$ is unique up to isomorphism.
	\end{proof}
	
	\paragraph{\textbf{Balanced differential 2-crossed modules}}
	Similarly, we can define a balanced differential 2-crossed module in order to develop the invariant form in the 3CS theory.
	\begin{definition}[Balanced differential 2-crossed modules]
		A differential 2-crossed module $(\mathcal{l}, \mathcal{h}, \mathcal{g};\beta, \alpha, \vartriangleright, \{, \})$ is said balanced if dim $\mathcal{l}$= dim $\mathcal{g}$.
	\end{definition}
	\begin{proposition}
		For any non balanced differential 2-crossed module $(\mathcal{l}, \mathcal{h}, \mathcal{g};\beta, \alpha, \vartriangleright, \{, \})$, there exists a balanced differential 2-crossed module $(\tilde{\mathcal{l}}, \tilde{\mathcal{h}}, \tilde{\mathcal{g}}; \tilde{\beta}, \tilde{\alpha}, \tilde{\vartriangleright}, \widetilde{\{, \}})$ minimally extending $(\mathcal{l}, \mathcal{h}, \mathcal{g};\beta, \alpha, \vartriangleright, \{, \})$.
	\end{proposition}
	\begin{proof}
		Firstly,  let dim $\mathcal{g} < $ dim $\mathcal{l}$. Then, there is a balanced
		differential 2-crossed module $(\tilde{\mathcal{l}}, \tilde{\mathcal{h}}, \tilde{\mathcal{g}}; \tilde{\beta}, \tilde{\alpha}, \tilde{\vartriangleright}, \widetilde{\{, \}})$ where
		\begin{itemize}
			\item[1.] $\tilde{\mathcal{l}}=\mathcal{l}$ and $\tilde{\mathcal{h}}=\mathcal{h}$;
			\item[2.] $\tilde{\mathcal{g}}=\mathcal{g} \oplus \mathcal{w}$, where $\mathcal{w}$ is a vector space such that dim $\mathcal{w}$ $=$ dim $\mathcal{l}$ $-$ dim $\mathcal{g}$, and the Lie bracket $[-, -]$ in the Lie algebra $\tilde{\mathcal{g}}$ is given by 
			\begin{align}
				[X_1 \oplus W_1, X_2 \oplus W_2] = [ X_1, X_2] \oplus 0;
			\end{align}
			\item[3.] the map $\tilde{\beta}= \beta$;
			\item[4.] the map $\tilde{\alpha}: \tilde{\mathcal{h}} \longrightarrow \tilde{\mathcal{g}}$ is given by
			\begin{align}
				\tilde{\alpha}(Y)= \alpha(Y) \oplus 0 \ \ \ Y \in \mathcal{h};
			\end{align}
			\item[5.] the action $\tilde{\vartriangleright}$ of $\tilde{\mathcal{g}}$ on $ \tilde{\mathcal{h}}$ and $ \tilde{\mathcal{l}}$ are given by 
			\begin{align}
				(X \oplus W)\tilde{\vartriangleright} Y = X \vartriangleright Y, \\
				(X \oplus W)\tilde{\vartriangleright} Z = X \vartriangleright Z;
			\end{align}
			\item[6.] the map $\widetilde{\{, \}} = \{, \}$,
		\end{itemize}
		for any $X, X_1, X_2 \in \mathcal{g}, Y \in \mathcal{h}, Z \in \mathcal{l}, W, W_1, W_2 \in \mathcal{w}$.

		Secondly, let dim $\mathcal{g} > $ dim $\mathcal{l}$. Then, there is a balanced differential crossed module  $(\mathcal{l}, \mathcal{h}, \mathcal{g};\beta, \alpha, \vartriangleright, \{, \})$ where 
		\begin{itemize}
			\item[1.] $\tilde{\mathcal{h}}=\mathcal{h}$ and $\tilde{\mathcal{g}}=\mathcal{g}$;
			\item[2.] $\tilde{\mathcal{l}}=\mathcal{l} \oplus \mathcal{v}$, where $\mathcal{v}$ is a vector space such that dim $\mathcal{v}$ $=$ dim $\mathcal{g}$ $-$ dim $\mathcal{l}$, and the Lie bracket $[-, -]$ in the Lie algebra $\tilde{\mathcal{l}}$ is given by 
			\begin{align}
				[Y_1 \oplus V_1, Y_2 \oplus V_2] = [ Y_1, Y_2] \oplus 0;
			\end{align}
			\item[3.] the map $\tilde{\beta}: \tilde{\mathcal{l}} \longrightarrow \tilde{\mathcal{h}}$ is given by
			\begin{align}
				\tilde{\beta}(Z \oplus V)= \tilde{\beta}(Z);
			\end{align}
			\item[4.] $\tilde{\alpha}= \alpha$;
			\item[5.] the action $\tilde{\vartriangleright}$ of $\tilde{\mathcal{g}}$ on $ \tilde{\mathcal{l}}$ is given by 
			\begin{align}
				X\tilde{\vartriangleright} (Z \oplus V) = (X \vartriangleright Z)\oplus 0;
			\end{align}
			\item[6.] the map $\widetilde{\{, \}}: \tilde{h} \times  \tilde{h} \longrightarrow \tilde{l}$ is given by
			\begin{align}
				\widetilde{\{ Y_1, Y_2 \}}= \{ Y_1, Y_2 \} \oplus 0,
			\end{align}
		\end{itemize}
		for any $X \in \mathcal{g}, Y_1, Y_2 \in \mathcal{h}, Z \in \mathcal{l}$, and $V, V_1, V_2 \in \mathcal{v}$.
		
		Further, the differential 2-crossed module $(\tilde{\mathcal{l}}, \tilde{\mathcal{h}}, \tilde{\mathcal{g}}; \tilde{\beta}, \tilde{\alpha}, \tilde{\vartriangleright}, \widetilde{\{, \}})$  is unique up to isomorphism.
	\end{proof}
	
	Using the above results, we can always suppose that the differential (2-)crossed module we are dealing with is balanced. 
	We have established the existence of balance, and now we further develop the invariant forms on the differential (2-)crossed modules.

	\paragraph{\textbf{Invariant forms on the differential crossed modules}} The study of invariant forms on Lie 2-algebras was originally proposed in \cite{ESRZ, R.Z1} as a way of constructing the semistrict 4-d CS theory. In a similar manner, we will introduce an invariant form on the differential crossed module, which is equivalent to a strict Lie 2-algebra.
	
	\begin{definition}
		Given a Lie crossed module $(H, G; \alpha, \vartriangleright)$ and the associated differential crossed module $(\mathcal{h}, \mathcal{g}; \alpha, \vartriangleright)$, a $G$-invariant form on $(\mathcal{h}, \mathcal{g}; \alpha, \vartriangleright)$ is defined as a non singular bilinear form $\langle - , - \rangle_{\mathcal{g},\mathcal{h}}: \mathcal{g} \times \mathcal{h} \longrightarrow \mathbb{R}$, satisfying
		\begin{align}
			&	\langle [X_1, X_2], Y \rangle_{\mathcal{g},\mathcal{h}} = - \langle X_2, X_1 \vartriangleright Y\rangle_{\mathcal{g},\mathcal{h}},\label{XXY}\\
			&\langle\alpha(Y_1), Y_2 \rangle_{\mathcal{g},\mathcal{h}} = \langle \alpha(Y_2), Y_1\rangle_{\mathcal{g},\mathcal{h}} \label{XY}
		\end{align}
		for any $X_1, X_2 \in \mathcal{g}$, and $Y, Y_1, Y_2 \in \mathcal{h}$.
	\end{definition}
	The non singularity of $\langle -, - \rangle_{\mathcal{g}, \mathcal{h}}$ implies that $(\mathcal{h}, \mathcal{g}; \alpha, \vartriangleright)$ is balanced, i.e. dim $\mathcal{h}=$ dim $\mathcal{g}$. Besides, the form $\langle - , - \rangle_{\mathcal{g},\mathcal{h}}$ is $G$-invariant, i.e.
	\begin{align}\label{gin}
		\langle g X g^{-1}, g \vartriangleright Y\rangle_{\mathcal{g},\mathcal{h}} = \langle X, Y\rangle_{\mathcal{g},\mathcal{h}}
	\end{align}
	for any $g \in G, X \in \mathcal{g}$, and $Y \in \mathcal{h}$.

	\paragraph{\textbf{Invariant forms on the differential 2-crossed modules}} 
	Using a similar argument, we can define an invariant form on the differential 2-crossed module.
	\begin{definition}
		Given a 2-crossed module $(L, H, G; \beta, \alpha, \vartriangleright, \{, \})$ and the associated differential 2-crossed module $(\mathcal{l}, \mathcal{h}, \mathcal{g}; \beta, \alpha, \vartriangleright, \{, \})$, a $G$-invariant form in $(\mathcal{l}, \mathcal{h}, \mathcal{g}; \beta, \alpha, \vartriangleright, \{, \})$ consists of a pair of bilinear forms given by 
		\begin{itemize}
			\item an antisymmetric non-degenerate  bilinear form $\langle - , - \rangle_\mathcal{h}:  \mathcal{h} \times \mathcal{h} \longrightarrow \mathbb{R}$, satisfying
			\begin{align}
				&	\langle [Y,Y_1], Y_2 \rangle_\mathcal{h}=-\langle Y_1,[Y,Y_2]\rangle_\mathcal{h},\\
				&	\langle Y, X\vartriangleright Y_1\rangle_\mathcal{h}=\langle Y_1, X\vartriangleright Y \rangle_\mathcal{h};\label{YX}
			\end{align}
			\item a non singular bilinear form $\langle - , - \rangle_{\mathcal{g},\mathcal{l}}: \mathcal{g} \times \mathcal{l} \longrightarrow \mathbb{R}$, satisfying
			\begin{align}
				&	\langle[X_1, X_2], Z\rangle_{\mathcal{g},\mathcal{l}} = - \langle X_2, X_1 \vartriangleright  Z\rangle_{\mathcal{g},\mathcal{l}} \label{XZ},\\
				&\langle \alpha(Y), Z \rangle_{\mathcal{g},\mathcal{l}} = - \langle \beta(Z), Y\rangle_{\mathcal{h}},\label{YZ}\\
				&\langle X, \{ Y_1, Y_2\} \rangle_{\mathcal{g}, \mathcal{l}}=\frac{1}{2} \langle Y_2, X \vartriangleright Y_1 \rangle_{\mathcal{h}},\label{XYY}
			\end{align}
			for $X, X_1, X_2 \in \mathcal{g}$, $Y, Y_1, Y_2 \in \mathcal{h}$ and $Z \in \mathcal{l}$.
		\end{itemize}
		
	\end{definition}
	The non singularity of $\langle -, - \rangle_{\mathcal{g}, \mathcal{l}}$ implies that $(\mathcal{l}, \mathcal{h}, \mathcal{g}; \beta, \alpha, \vartriangleright, \{, \})$ is balanced, i.e. dim $\mathcal{l}=$ dim $\mathcal{g}$. Besides, 
	these forms are $G$-invariant, i.e. 
	\begin{align}
		\langle g\vartriangleright Y,  g\vartriangleright Y' \rangle_\mathcal{h}&= \langle Y,Y'\rangle_\mathcal{h},\\
		\langle g X g^{-1}, g \vartriangleright Z\rangle_{\mathcal{g},\mathcal{l}} &= \langle X, Z\rangle_{\mathcal{g},\mathcal{l}}
	\end{align}
	for any $g \in G, X \in \mathcal{g}, Y, Y'\in \mathcal{h}$ and $ Z \in \mathcal{l}$.

	\section{Generalized differential forms valued in differential (2-)crossed modules}\label{sub4}
	In this section, we first recall the generalized differential forms of type $N$ in the framework of GDC \cite{Robbinson0, Robinson1, Robinson2, Robinson3, Robbinson4, Robbinson5, Robinson6}. Then we formulate analogous results, which are applicable to the generalized differential forms valued in differential (2-)crossed modules.
	Since we aim to the construction of 2CS and 3CS gauge theories, we focus on just the types $N=1$ and $N=2$.

	\subsection{Generalized differential forms}\label{GDC1}
	There are $N$ linearly independent minus $1$-forms \{$\xi^i$\}, ($i = 1, 2, ... , N$), which assumed to satisfy the ordinary distributive and associative laws of exterior algebra. The exterior product rule is given by
	\begin{align}
		\xi^i \xi^j = - \xi^j \xi^i, \ \ \  \overset{p}{a} \, \xi^i = (-1)^p  \xi^i  \,\overset{p}{a} ,
	\end{align}
	where $ \overset{p}{a}$ is any ordinary $p$-form on a manifold $M$ of dimension $n$. 
	In particular, they satisfy a condition of linear independence, $\xi^1 \xi^2 ... \xi^N \neq 0 $ and $(\xi^i )^2 =0$.  In order to ensure that their exterior derivatives are zero-forms and that $d^2 = 0$, they are required to satisfy the condition $d \, \xi^i = k^i$, where $k^i$ is a constant for any $i=1, ..., N$.  
	
	A generalized $p$-form of type $N$ is thus defined as
	\begin{align}\label{gpf}
		\mathcal{W} = \overset{p}{a} +  \overset{p+1}{a}_{i_1} \, \xi^{i_1} + \dfrac{1}{2!} \overset{p+2}{a}_{i_1 i_2} \, \xi^{i_1}\xi^{i_2} + ... + \dfrac{1}{j!} \overset{p+j}{a}_{i_1 ... i_j} \, \xi^{i_1} ... \xi^{i_j} +... + \dfrac{1}{N!} \overset{p+N}{a}_{i_1 ... i_N} \, \xi^{i_1} ... \xi^{i_N},
	\end{align}
	where $\overset{p}{a}$, $\overset{p+1}{a}_{i_1}$, ..., $\overset{p+j}{a}_{i_1 ... i_j} = \overset{p+j}{a}_{[i_1 ... i_j]}$, ..., $\overset{p+N}{a}_{i_1 ... i_N} $ are, respectively, ordinary $p$-, $(p+1)$-, ..., $(p+j)$-, ..., $(p+N)$- ordinary forms; $-N \leqslant p \leqslant n$, $j = 1, 2, ..., N$ and $i_1, ..., i_j, ..., i_N$ range and sum over $1, 2, ..., N$.
	It then follows that the generalized forms
	satisfy the same basic rules
	of exterior multiplication and differentiation as those which govern the algebra
	and calculus of ordinary differential forms, apart from $p$ taking the value minus one in standard formulae.
	For example, these exterior products and derivatives of generalized $p$-form $\mathcal{U}$ and $q$-form $\mathcal{V}$ satisfy the standard rules of exterior algebra
	\begin{align}
		\mathcal{U}\boldsymbol{\wedge} \mathcal{V} &= (-1)^{pq} \mathcal{V} \boldsymbol{\wedge} \mathcal{U},\\
		\textbf{d} (\mathcal{U}\boldsymbol{\wedge} \mathcal{V})&=d \mathcal{U} \wedge \mathcal{V} + (-1)^p \mathcal{U} \wedge \mathcal{V},
	\end{align}
	and $\textbf{d}^2=0$.
	From \eqref{gpf}, we can see that a generalized $p$-form of type $N = 0$ is an ordinary differential $p$-form, and the exterior algebra of type $N = 0$ forms is the ordinary exterior algebra.
	Further discussion of type $N$ generalized forms can be found in \cite{Robinson1}.
	In this paper, we focus on the type $N=1$ and $N=2$ generalized $p$-forms.
	
	\textbf{Type $N=1$ generalized $p$-forms.}
	As in the equation \eqref{gpf}, the type $N=1$ generalized $p$-form is defined as
	\begin{align}\label{U}
		\mathcal{U} = \overset{p}{a} +  \overset{p+1}{a}\xi,
	\end{align}
	where $\overset{p}{a}$ and $\overset{p+1}{a}$ are, respectively, ordinary $p$- and $(p+1)$-forms and $p$ can take integer values from $-1$ to $n$ with $\overset{-1}{a} = 0$.
	Then the generalized exterior product of the generalized $p$-form $\mathcal{U}$ in \eqref{U} and a generalized $q$-form $\mathcal{V}= \overset{q}{b} +  \overset{q+1}{b}\xi$ is given by
	\begin{align}\label{gep}
		\mathcal{U} \boldsymbol{\wedge} \mathcal{V} = \overset{p}{a} \wedge \overset{q}{b}+ ( \overset{p}{a} \wedge \overset{q+1}{b} + (-1)^q \,\overset{p+1}{a} \wedge  \overset{q}{b}) \xi,
	\end{align}
	and the generalized exterior derivative of $\mathcal{U}$ is given by
	\begin{align}\label{gd}
		\textbf{d} \mathcal{U} = d \, \overset{p}{a} + (-1)^{p+1} k \, \overset{p+1}{a} +  d \, \overset{p+1}{a} \xi,
	\end{align}
	where $k$ is a constant.	   
	
	\textbf{Type $N=2$ generalized $p$-forms.}
	As in the equation \eqref{gpf}, the type $N=2$ generalized $p$-form is defined as
	\begin{align}\label{UU}
		\mathcal{U} = \overset{p}{a} +  \overset{p+1}{a}_i \xi^i + \overset{p+2}{a}\xi^{12}, \ \ \ i =1, 2,
	\end{align}
	where $\overset{p}{a}$, $\overset{p+1}{a}_i$ and $\overset{p+2}{a}$ are, respectively, ordinary $p$-, $(p+1)$- and $(p+2)$-forms, and $p$ can take integer values from $-2$ to $n$ with $\overset{-1}{a} =\overset{-2}{a} = 0$ and $\xi^{12}=\xi^1 \wedge \xi^2$.
	Then the generalized exterior product of the generalized $p$-form $\mathcal{U}$ in \eqref{UU} and a generalized $q$-form $\mathcal{V} = \overset{q}{b} +  \overset{q+1}{b}_i \xi^i + \overset{q+2}{b}\xi^{12}$ is given by
	\begin{align}\label{gepq}
		\mathcal{U} \boldsymbol{\wedge} \mathcal{V} =& \overset{p}{a} \wedge \overset{q}{b}+ ( \overset{p}{a} \wedge \overset{q+1}{b}_i + (-1)^q \,\overset{p+1}{a}_i \wedge  \overset{q}{b}) \xi^i \nonumber\\
		&+ (\overset{p}{a}\wedge \overset{q+2}{b} + (-1)^{q+1}(\overset{p+1}{a}_1\wedge \overset{q+1}{b}_2 - \overset{p+1}{a}_2\wedge \overset{q+1}{b}_1) + \overset{p+2}{a}\wedge \overset{q}{b}) \xi^{12},
	\end{align}
	and the generalized exterior derivative of $\mathcal{U}$ is given by
	\begin{align}\label{gdq}
		\textbf{d} \mathcal{U} = d \, \overset{p}{a} + (-1)^{p+1} k^i \overset{p+1}{a}_i +  (d \overset{p+1}{a}_1 +(-1)^{p+1}k^2 \overset{p+2}{a})\xi^1 + (d \overset{p+1}{a}_2 +(-1)^{p}k^1 \overset{p+2}{a})\xi^2 + d \overset{p+2}{a}\xi^{12} ,
	\end{align}
	with $k^i$ for $i=1,2$.

	\subsection{Type $N=1$ generalized forms valued in differential crossed modules}\label{1gf}
	
	Similar to ordinary algebra-valued differential forms, we define a generalized form valued in the differential crossed module, which consists of a pair of ordinary algebra-valued differential forms.
	
	\begin{definition}
		Given a differential crossed module $(\mathcal{h}, \mathcal{g}; \alpha, \vartriangleright)$ and an $n$-dimension manifold $M$, a type $N=1$ generalized $p$-form valued in $(\mathcal{h}, \mathcal{g}; \alpha, \vartriangleright)$ is given by 
		\begin{align}\label{A}
			\mathcal{A} = \overset{p}{A} +  \overset{p+1}{A}\xi, \ \ \ -1\leq p \leq n,
		\end{align}
		with $\overset{p}{A} \in \Lambda^p(M, \mathcal{g})$ and $\overset{p+1}{A} \in \Lambda^{p+1}(M, \mathcal{h})$.
	\end{definition}
	
	According to the GDC in subsection \ref{GDC1}  and the higher gauge theory in the derived formulation considered by Zucchini in \cite{ESRZ, R.Z, R.Z2, R.Z3}, we can define the generalized exterior product and derivative for the above generalized $p$-forms valued in the differential crossed modules.
	
	\begin{definition}
		Given a generalized $p$-form $\mathcal{A}$ and a $q$-form $\mathcal{B}$ valued in  a differential crossed module $(\mathcal{h}, \mathcal{g}; \alpha, \vartriangleright)$, 
		\begin{align}
			\mathcal{A} = \overset{p}{A} +  \overset{p+1}{A}\xi,\ \ \ \mathcal{B} = \overset{q}{B} +  \overset{q+1}{B}\xi \label{B}
		\end{align}
		with $\overset{p}{A} \in \Lambda^p(M, \mathcal{g})$, $\overset{p+1}{A} \in \Lambda^{p+1}(M, \mathcal{h})$, $\overset{q}{B} \in \Lambda^q(M, \mathcal{g})$ and $\overset{q+1}{B} \in \Lambda^{q+1}(M, \mathcal{h})$,
		the generalized exterior product is given by
		\begin{align}\label{N1}
			\mathcal{A} \boldsymbol{\wedge} \mathcal{B} = \overset{p}{A} \wedge \overset{q}{B}+ \overset{p}{A} \wedge^{\vartriangleright} \overset{q+1}{B}  \xi,
		\end{align}
		and the generalized exterior derivative of $\mathcal{A}$ is given by
		\begin{align}\label{N11}
			\textbf{d} \mathcal{A} = d \, \overset{p}{A} + (-1)^{p+1} k \alpha(\overset{p+1}{A}) +  d \, \overset{p+1}{A} \xi,
		\end{align}
		where $k$ is a constant. 
	\end{definition}
	The important point to note here is that the equations \eqref{N1} and \eqref{N11} are induced by \eqref{gep} and \eqref{gd}, respectively.  This definition provides a computing method for higher algebra-valued differential forms.
	It is easy to show that 
	\begin{align}\label{N111}
		\mathcal{A}\wedge^{[, ]} \mathcal{B}=\overset{p}{A}\wedge^{[, ]}\overset{q}{B} +(\overset{p}{A}\wedge^{\vartriangleright}\overset{q+1}{B} - (-1)^{pq}\overset{q}{B}\wedge^{\vartriangleright}\overset{p+1}{A})\xi.
	\end{align}
	Furthermore, it can be straightforwardly verified that
	\begin{align}
		\textbf{d}(\mathcal{A}\wedge^{[, ]} \mathcal{B})= \textbf{d} \mathcal{A}\wedge^{[, ]} \mathcal{B} + (-1)^p \mathcal{A} \wedge^{[, ]} \textbf{d} \mathcal{B}
	\end{align}
	and that
	\begin{align}
		\textbf{d}^2=0.
	\end{align}
	
	Encouraged by these constructions of generalized differential form, we can define a generalized bilinear forms $ \ll -, - \gg$ for the higher algebra-valued differential forms.
	Given a pair of generalized forms of type $N=1$, $(\mathcal{A}, \mathcal{B})$ of the term  \eqref{B},  define 
	\begin{align}\label{gb}
		\ll \mathcal{A}, \mathcal{B} \gg=\ll \overset{p}{a} +  \overset{p+1}{a}\xi, 	\overset{q}{b} +  \overset{q+1}{b}\xi\gg=\langle  \overset{p}{a}, \overset{q+1}{b}\rangle_{\mathcal{g}, \mathcal{h}} + \langle \overset{q}{b},  \overset{p+1}{a} \rangle_{\mathcal{g}, \mathcal{h}}.
	\end{align}
	From  the property of the  non singular $G$-invariant bilinear form on $(\mathcal{h}, \mathcal{g}; \alpha, \vartriangleright)$, one can apparently note that $\ll -, - \gg$ is  a symmetric  $G$-invariant bilinear form.
	This structure will lead us to some unexpected and novel results in the construction of 2CS theory.

	\subsection{Type $N=2$ generalized forms valued in differential $2$-crossed modules}\label{2gf}
	We have already defined the differential crossed module-valued generalized forms of type $N=1$, and along this approach we will establish a type $N=2$ generalized forms valued differential $2$-crossed modules.
	
	\begin{definition}
		Given a differential $2$-crossed module $(\mathcal{l}, \mathcal{h}, \mathcal{g}; \beta, \alpha, \vartriangleright, \{, \})$ and a $n$-dimension manifold $M$, a type $N=2$ generalized $p$-form valued in $(\mathcal{l}, \mathcal{h}, \mathcal{g}; \beta, \alpha, \vartriangleright, \{, \})$ is given by 
		\begin{align}\label{AA}
			\mathcal{A} = \overset{p}{a} +  \overset{p+1}{a}_i \xi^i + \overset{p+2}{a}\xi^{12}, \ \ \ -2\leq p \leq n
		\end{align}
		with $i =1, 2$,
		$\overset{p}{a} \in \Lambda^p(M, \mathcal{g})$, $\overset{p+1}{a}_i \in \Lambda^{p+1}(M, \mathcal{h})$ and $\overset{p+2}{a} \in \Lambda^{p+2}(M, \mathcal{l})$.
	\end{definition}
	
	In the same measure, we can define the generalized exterior product and derivative for the above generalized $p$-forms valued in the differential 2-crossed modules.
	\begin{definition}
		Given a generalized $p$-form $\mathcal{A}$ and a $q$-form $\mathcal{B}$ valued in a differential 2-crossed module $(\mathcal{l}, \mathcal{h}, \mathcal{g}; \beta, \alpha, \vartriangleright, \{, \})$,
		\begin{align}\label{BB}
			\mathcal{A} = \overset{p}{a} +  \overset{p+1}{a}_i \xi^i + \overset{p+2}{a}\xi^{12},\nonumber\\
			\mathcal{B} = \overset{q}{b} +  \overset{q+1}{b}_i \xi^i + \overset{q+2}{b}\xi^{12},
		\end{align}
		with $\overset{q}{b} \in \Lambda^q(M, \mathcal{g})$, $\overset{q+1}{b}_i \in \Lambda^{q+1}(M, \mathcal{h})$ and $\overset{q+2}{b} \in \Lambda^{q+2}(M, \mathcal{l})$ for $i=1, 2$,
		the generalized exterior product is given by
		\begin{align}\label{N2}
			\mathcal{A} \boldsymbol{\wedge} \mathcal{B} = \overset{p}{a} \wedge \overset{q}{b}+ \overset{p}{a} \wedge^{\vartriangleright} \overset{q+1}{b}_i  \xi^i + (\overset{p}{a}\wedge^{\vartriangleright} \overset{q+2}{b} + (-1)^{q+1}\overset{p+1}{a}_1\wedge^{\{, \}} \overset{q+1}{b}_2) \xi^{12},
		\end{align}
		which is induced by \eqref{gepq},
		and the generalized exterior derivative of $\mathcal{A}$ is given by
		\begin{align}\label{N22}
			\textbf{d} \mathcal{A} = d \, \overset{p}{a} + (-1)^{p+1} k^i \alpha(\overset{p+1}{a}_i )+  (d \overset{p+1}{a}_i +(-1)^{p+i}k^j \beta(\overset{p+2}{a}))\xi^i + d \overset{p+2}{a}\xi^{12} ,
		\end{align}
		which is induced by \eqref{gdq}, and where $k^i$ is a constant and $i, j =1, 2, i\neq j$.   
	\end{definition}
	
	In this definition, the equations \eqref{N2} and \eqref{N22} are induced by \eqref{gepq} and \eqref{gdq}, respectively. 
	Consequently, we infer that
	\begin{align}\label{N222}
		\mathcal{A}\wedge^{[, ]} \mathcal{B}=&\overset{p}{a}\wedge^{[, ]}\overset{q}{b} +(\overset{p}{a}\wedge^{\vartriangleright}\overset{q+1}{b}_i - (-1)^{pq}\overset{q}{b}\wedge^{\vartriangleright}\overset{p+1}{a}_i)\xi^i + (\overset{p}{a}\wedge^{\vartriangleright}\overset{q+2}{b} - (-1)^{pq}\overset{q}{b}\wedge^{\vartriangleright}\overset{p+2}{a}\nonumber\\
		& +(-1)^{q+1}\overset{p+1}{a}_1 \wedge^{\{, \}}\overset{q+1}{b}_2 - (-1)^{pq +p+1}\overset{q+1}{b}_1 \wedge^{\{, \}}\overset{p+1}{a}_2)\xi^{12},\\
		\textbf{d}(\mathcal{A}\wedge^{[, ]} \mathcal{B})&= \textbf{d} \mathcal{A}\wedge^{[, ]} \mathcal{B} + (-1)^p \mathcal{A} \wedge^{[, ]} \textbf{d} \mathcal{B},
	\end{align}
	and that
	\begin{align}
		\textbf{d}^2=0.
	\end{align}
	
	By the same token, we define a generalized bilinear form $\ll -, -\gg$ for these differential 2-crossed module-valued generalized forms. 
	Given a pair of generalized forms of type $N=2$, $(\mathcal{A}, \mathcal{B})$, of the term \eqref{BB}, we define 
	\begin{align}\label{gl3cs}
		\ll \mathcal{A}, \mathcal{B} \gg=&\ll \overset{p}{a} +  \overset{p+1}{a}_i \xi^i + \overset{p+2}{a}\xi^{12}, 	\overset{q}{b} +  \overset{q+1}{b}_j \xi^j + \overset{q+2}{b}\xi^{12}\gg\nonumber\\
		=&\langle  \overset{p}{a}, \overset{q+2}{b}\rangle_{\mathcal{g}, \mathcal{l}} + \langle \overset{q}{b},  \overset{p+2}{a} \rangle_{\mathcal{g}, \mathcal{l}} - k^i \langle \overset{p+1}{a}_i, \overset{q+1}{b}_j \rangle_{\mathcal{h}},\ \ \ \ \ \  i, j=1, 2, i \neq j.
	\end{align}
	From  the property of the  non singular $G$-invariant bilinear form on $(\mathcal{l}, \mathcal{h}, \mathcal{g}; \beta, \alpha, \vartriangleright, \{, \})$, one can also apparently note that $\ll -, - \gg$ is  a $G$-invariant bilinear form.
	This structure will be useful in the construction of 3CS theory.
	
	\section{Generalized connections}\label{sub5}
	In this section, we present a generalized formulation of the higher gauge theory based on the subsections \ref{1gf} and \ref{2gf}. Under the framework of the GDC, it is easy to show the close relationship of higher to ordinary gauge theory.  Based on the above argument, it allows so to many import ideas and techniques of the $CS$ gauge theory to the higher $CS$ gauge theories in the next section. 
	
	\paragraph{\textbf{Type $N=1$ generalized connections}}
	Given a $2$-connection $(A, B)$,
	the generalized connection $1$-form of type $N=1$ is given by 
	\begin{align}
		\mathcal{A}= A + B \xi.
	\end{align}
	Then the generalized curvature $2$-form of $\mathcal{A}$ is $\mathcal{F}= \textbf{d}\mathcal{A} + \mathcal{A}\wedge \mathcal{A}$.   
	By using \eqref{N1} and \eqref{N11} with $k=-1$, a straightforward computation gives
	\begin{align}
		\mathcal{F}= dA + A \wedge A - \alpha(B)+(dB + A \wedge^{\vartriangleright}B)\xi= \Omega_1 + \Omega_2\xi
	\end{align}
	with the $2$-curvature $(\Omega_1, \Omega_2)$ of the form \eqref{1}.
	
	The generalized curvature $\mathcal{F}$ also satisfies the generalized Bianchi Identity
	\begin{align}
		d  \mathcal{F} + \mathcal{A}\wedge^{[,]}\mathcal{F}=0,
	\end{align}
	which gives
	\begin{align}
		d \Omega_1 + A \wedge^{[, ]}\Omega_1 + \alpha(\Omega_2)+(d \Omega_2 + A \wedge^{\vartriangleright}\Omega_2 - \Omega_1 \wedge^{\vartriangleright}B)\xi=0
	\end{align}
	with the 2-curvature $(\Omega_1, \Omega_2)$, i.e. the 2-Bianchi identities  \eqref{2BI}.

	\paragraph{\textbf{Type $N=2$ generalized connections}}
	Given a $3$-connection $(A, B,C)$,
	the generalized connection $1$-form of type $N=2$ is given by 
	\begin{align}
		\mathcal{A}= A + B \xi^1+ B \xi^2 + C \xi^{12}.
	\end{align}
	Then the generalized curvature $2$-form of $\mathcal{A}$ is $\mathcal{F}= \textbf{d}\mathcal{A} + \mathcal{A}\wedge \mathcal{A}$.   By using \eqref{N2} and \eqref{N22} with $k^1=0$ and $k^2=-1$, a straightforward computation gives
	\begin{align}
		\mathcal{F}= &dA + A \wedge A - \alpha(B) +(dB + A \wedge^{\vartriangleright}B - \beta(C))\xi^1 \nonumber\\
		& + (dB + A \wedge^{\vartriangleright}B)\xi^2 + (dC + A \wedge^{\vartriangleright}C + B \wedge^{\{, \}}B)\xi^{12}\nonumber\\
		=&\Omega_1 + \Omega_2 \xi^1 +(\Omega_2 + \beta(C))\xi^2 + \Omega_3 \xi^{12}.
	\end{align}
	
	Similarly, there is the generalized Bianchi Identity
	\begin{align}
		d  \mathcal{F} + \mathcal{A}\wedge^{[,]}\mathcal{F}=0,
	\end{align}
	which gives
	\begin{align}\label{gbi}
		&	d \Omega_1 + A\wedge^{[, ]}\Omega_1 + \alpha(\Omega_2) +(d \Omega_2 + A\wedge^{\vartriangleright} \Omega_2+ \beta(\Omega_3)-\Omega_1 \wedge^{\vartriangleright}B)\xi^1\nonumber\\
		& +(d\Omega_2 + A\wedge^{\vartriangleright}(\Omega_2 + \beta(C))-\Omega_1 \wedge^{\vartriangleright}B +\beta(dC))\xi^2\nonumber\\
		&+(d \Omega_3 + A \wedge^{\vartriangleright}\Omega_3 - \Omega_1 \wedge^{\vartriangleright}C - B\wedge^{\{, \}}\Omega_2- \Omega_2\wedge^{\{, \}}B)\xi^{12}=0,
	\end{align}
	with the 3-curvature $(\Omega_1, \Omega_2, \Omega_3)$, i.e. the $3$-Bianchi Identity \eqref{3bi}
	\begin{align}
		d \Omega_1 + A \wedge^{[, ]} \Omega_1 + \alpha(\Omega_2)&=0,\label{BI1}\\
		d \Omega_2 + A \wedge^{\vartriangleright}\Omega_2 - \Omega_1 \wedge^{\vartriangleright}B + \beta(\Omega_3)&=0, \label{BI2}\\
		d \Omega_3 + A \wedge^{\vartriangleright}\Omega_3 - \Omega_1\wedge^{\vartriangleright}C - B\wedge^{\{, \}}\Omega_2 - \Omega_2 \wedge^{\{, \}}B &=0, \label{BI3}
	\end{align}
	According to \eqref{BI1} and \eqref{BI2}, it follows that the coefficient of $\xi^2$ will vanish in equation \eqref{gbi}.

	\section{Higher Chern-Simons theory}\label{sub6}
	In this section, we first briefly remind the ordinary CS gauge theory and the relevant Chern-Weil theorem.
	Then we will construct the 2CS and 3CS gauge theories based on the generalized connections, and generalize these associated discussions in the ordinary CS theory to the HCS.
	\subsection{CS gauge theory and Chern-Weil theorem}\label{sub61}
	Let us consider a principle bundle $P(M, G)$ on a 3-dimensional manifold $M$ with a gauge group $G$ and the associated Lie algebra $\mathcal{g}$. Locally, a gauge field is a connection 1-form $A \in \Lambda^1(M, \mathcal{g})$ that transforms under the gauge transformation like
	\begin{align}\label{423}
		A'=g^{-1}A g + g^{-1}dg,
	\end{align}
	so that the corresponding curvature 2-form $F= dA + A \wedge A  \in \Lambda^2(M, \mathcal{g})$ transforms like
	\begin{align}
		F'= g^{-1}F g.
	\end{align}
	Besides, they satisfy the Bianchi Identity, $d F + A \wedge^{[, ]}F =0$.
	
	In order to consider their topological properties,  let
	us first briefly remind the curvature invariant symmetric polynomial on the bundle $P(M, G)$, which is called the \textbf{second Chern form},
	\begin{align}
		P(F,F)= \langle F, F\rangle_{\mathcal{g}},
	\end{align}
	where $\langle-, -\rangle_{\mathcal{g}}$ is a $G$-invariant symmetric non-degenerate bilinear form, satisfying the invariance condition
	\begin{align}\label{Sg}
		\langle [X, X'], X'' \rangle_{\mathcal{g}} = -\langle X', [X, X'']\rangle_{\mathcal{g}}
	\end{align}
	for any $X, X', X'' \in \mathcal{g}$. Under the gauge transformation \eqref{423}, it is not difficult to verify that the second Chern form is gauge invariant
	\begin{align}\label{FF0}
		P(F', F')=P(F, F).
	\end{align}

	Take an action of the form
	\begin{align}
		S=\int_{M} \langle F, F\rangle_{\mathcal{g}},
	\end{align}
	and the action is gauge invariant due to \eqref{FF0}.  However, the action gives completely trivial equations. To see this, simply compute, having
	\begin{align}
		\delta S=&\int_{M}2 \langle \delta F, F \rangle_{\mathcal{g}}\nonumber\\
		=&\int_{M}2 \langle  \delta dA + A \wedge^{[, ]}\delta A, F \rangle_{\mathcal{g}}\nonumber\\
		=&\int_{M}2 \langle \delta A, dF + A \wedge^{[, ]}F \rangle_{\mathcal{g}}.
	\end{align}
	It is apparent from the above equation that $\delta S =0$ for all $A$.
	
	Next we recall the Chern-Weil theorem for the second Chern form \cite{GHY} and give a rigorous proof. This theorem will be generalized to the 2CS and 3CS gauge theories in subsection \ref{sub62} and \ref{sub63}, respectively. And these proofs of the new Chern-Weil theorems can be completed by the method analogous to that used as follows.
	\begin{theorem}[Chern-Weil theorem]
		The second Chern form $P(F, F)$ satisfies
		\begin{itemize}\label{theorem1}
			\item[1).] $P(F, F)$ is closed, i.e. $d P(F)=0$;
			\item[2).] $P(F,F)$ has topologically invariant integrals, namely satisfying the Chern-Weil homomorphism formula:
			\begin{align}\label{F01}
				P(F^1, F^1) - P(F^0, F^0) = d Q(A^0, A^1),
			\end{align}
			with
			\begin{align}
				Q(A^0, A^1)= 2\int_{0}^{1} P(A^1 - A^0, F^t)dt,
			\end{align}
			where $A^0$ and $A^1$ are two connection 1-forms, $F^0$ and $F^1$ the corresponding curvature 2-forms,
			\begin{align}
				A^t = A^0 + t \eta, \ \ \ \eta = A^1 - A^0, \ \ \ (0 \leq t \leq 1),
			\end{align}
			the interpolation between $A^0$ and $A^1$,
			\begin{align}
				F^t= dA^t + A^t \wedge A^t,
			\end{align}
			the curvature of this interpolation.
		\end{itemize}
	\end{theorem}
	\begin{proof}
		\begin{itemize}
			\item[1).] According to the Bianchi Identity, have
			\begin{align}
				d P(F, F)=&\langle d F, F \rangle_{\mathcal{g}} + \langle F, dF \rangle_{\mathcal{g}}\nonumber\\
				=&-\langle A \wedge^{[, ]} F, F \rangle_{\mathcal{g}} - \langle F, A \wedge^{[, ]}F \rangle_{\mathcal{g}}\nonumber\\
				=&0
			\end{align}
			by using \eqref{Sg}.
			\item[2).] It is easy to see that
			\begin{align}
				\frac{d}{dt}F^t = d \eta + A^t \wedge^{[, ]}\eta \equiv D^t\eta.
			\end{align}
			Hence 
			\begin{align}
				& \frac{d}{dt}P(F^t, F^t)= 2\langle \frac{d}{dt}F^t, F^t \rangle_{\mathcal{g}}\nonumber\\
				& =2\langle D^t\eta, F^t \rangle_{\mathcal{g}}
				=2d\langle \eta, F^t \rangle_{\mathcal{g}}= 2 d P(\eta, F^t),
			\end{align}
			by using the Bianchi Identity $D^t F^t =0$. Thus 
			\begin{align}
				P(F^1, F^1) - P(F^0, F^0) &= 2d \int_{0}^{1} P(\eta, F^t)dt\nonumber\\
				&=d Q(A^0, A^1).
			\end{align}
		\end{itemize}
		
	\end{proof}
	
	This theorem shows that $P(F^1, F^1)$ and $P(F^0, F^0)$ differ by an exact form. In other wards, their integrals over 4-d manifolds without boundary give the same results, and $Q(A^0, A^1)$ is called the secondary topological class.
	
	In particular, we consider $A^1=A$ and $A^0=0$ in \eqref{F01}, then
	\begin{align}
		\langle F, F\rangle_{\mathcal{g}}= d Q_{CS},
	\end{align}
	where $Q_{CS}$  is given by
	\begin{align}
		Q_{CS}= 2 \int_{0}^{1} \langle A, t dA + t^2 A \wedge A \rangle_{\mathcal{g}} dt= \langle A, d A + \frac{1}{3}A \wedge^{[, ]}A \rangle_{\mathcal{g}},
	\end{align}
	which is called the \textbf{Chern-Simons form}.
	Then the Chern-Simons action  on a general 3-dimensional manifold $M$ is defined by
	\begin{align}
		S_{CS}&= \frac{\kappa}{4 \pi} \int_{M}\langle A, d A + \frac{1}{3}A \wedge^{[, ]}A \rangle_{\mathcal{g}}\nonumber\\
		&=\frac{\kappa}{4 \pi} \int_{M}\langle A, F - \frac{1}{3}A \wedge A \rangle_{\mathcal{g}},
	\end{align}
	with $\kappa$  is the coupling constant, sometimes called the level of the theory. This action is topological because it does not depend on a choice of metric on $M$. It is easy to show that the equation of motion is the condition of flatness for the connection:
	\begin{align}
		F = d A + A \wedge A = 0.
	\end{align}
	
	\subsection{2CS gauge theory and 2-Chern-Weil theorem}\label{sub62}
	In this section, we construct a 2CS gauge theory based on the type $N=1$ generalized connection theory. This construction was inspired by \cite{R.Z, R.Z1}. 
	As ordinary CS gauge theory exists only in
	odd dimensional manifolds, the 2CS gauge
	theory exists in $4$-dimensional manifolds and is
	built in the framework of  the $2$-gauge theory \cite{Baez.2010}.
	The $2$-gauge symmetry of this
	theory is encoded by a Lie crossed module.

	Consider a Lie crossed module $O= (H,G; \alpha, \vartriangleright)$ and  the associated differential crossed module $\mathcal{O}= (\mathcal{h},\mathcal{g}; \alpha, \vartriangleright)$. We assume that $M$ is an oriented, compact manifold without boundary. No further restrictions are imposed. Let $(A, B)$ be a 2-connection on the principle 2-bundle $P(M, O)$. The 2CS action is defined by a type $N=1$ generalized connection $\mathcal{A}= A + B \xi$ as follows
	\begin{align}\label{CS4}
		S_{2CS}=\frac{k}{4 \pi}\int \ll \mathcal{A}, d \mathcal{A} + \frac{1}{3}\mathcal{A}\wedge^{[,]}\mathcal{A}\gg.
	\end{align}
	By using \eqref{N1}, \eqref{N111} with $k=-1$ and \eqref{gb}, a straightforward computation gives
	\begin{align}\label{2cs}
		S_{2CS}=&\frac{\kappa}{4 \pi}\int \ll A + B \xi, dA + \frac{2}{3}A \wedge A - \alpha(B) + (d B + \frac{2}{3}A\wedge^{\vartriangleright}B)\xi\gg\nonumber\\
		=&\frac{k}{4 \pi}\int \langle A, d B + \frac{2}{3}A\wedge^{\vartriangleright}B\rangle_{\mathcal{g}, \mathcal{h}} + \langle dA + \frac{2}{3}A \wedge A - \alpha(B), B \rangle_{\mathcal{g}, \mathcal{h}}.
	\end{align}
	Using \eqref{XXY}, we have
	\begin{align}\label{2cs0}
		\langle A, A\wedge^{\vartriangleright}B\rangle_{\mathcal{g}, \mathcal{h}}=\langle A \wedge^{[, ]}A, B\rangle_{\mathcal{g}, \mathcal{h}}.
	\end{align}
	By substituting \eqref{2cs0} into \eqref{2cs}, we obtain 
	\begin{align}
		S_{2CS}=&\frac{k}{4 \pi}\int \langle A, dB + \frac{1}{2} A \wedge^{\vartriangleright}B \rangle_{\mathcal{g}, \mathcal{h}} + \langle dA + A \wedge A - \alpha(B), B \rangle_{\mathcal{g}, \mathcal{h}}\nonumber\\
		=&\frac{k}{4 \pi}\int \langle 2 F - \alpha(B), B \rangle_{\mathcal{g}, \mathcal{h}},
	\end{align}
	with $F= dA + A \wedge A$. So, the 2CS theory can be described as a generalized BF theory with a cosmological term determined by the map $\alpha$. 
	
	Varying the action $S_{2CS}$ with respect to $A$ and $B$, we obtain the equations of motion:
	\begin{align}
		&\delta A: d B + A \wedge^{\vartriangleright}B = \Omega_2 =0,\\
		&\delta B: F - \alpha(B)= \Omega_1= 0.
	\end{align}
	Clearly, the solution of these equations gives a flat $2$-connection $(A, B)$ on the trivial principal $O$-$2$-bundle over $M$, analogously to the standard $CS$ theory.

	Let us generalize the second Chern form \eqref{FF0} to a new one with the 2-curvature $(\Omega_1, \Omega_2)$,
	\begin{align}
		P(\Omega_1, \Omega_2)= 2\langle \Omega_1, \Omega_2\rangle_{\mathcal{g}, \mathcal{h}},
	\end{align}
	which is called the \textbf{second 2-Chern form}. 
	
	\begin{lemma}\label{lemma1}
		$P(\Omega_1, \Omega_2)= 2\langle \Omega_1, \Omega_2\rangle_{\mathcal{g}, \mathcal{h}}$ is 2-gauge invariant under the general 2-gauge transformation \eqref{33}. 
	\end{lemma}
	\begin{proof}
		\begin{align}
			P(\Omega''_1, \Omega''_2)&=  2\langle \Omega''_1, \Omega''_2\rangle_{\mathcal{g}, \mathcal{h}} \nonumber\\
			&=2 \langle g^{-1} \Omega_1 g, g^{-1} \vartriangleright \Omega_2 + (g^{-1} \Omega_1 g)\wedge^{\vartriangleright}\phi \rangle_{\mathcal{g}, \mathcal{h}}\nonumber\\
			&=2\langle \Omega_1, \Omega_2 \rangle_{\mathcal{g}, \mathcal{h}} + 2\langle \Omega_1, \Omega_1 \wedge^{\vartriangleright}\phi \rangle_{\mathcal{g}, \mathcal{h}}\nonumber\\
			&=2\langle \Omega_1, \Omega_2 \rangle_{\mathcal{g}, \mathcal{h}},
		\end{align}
		by using the invariance of $\langle-, - \rangle_{\mathcal{g}, \mathcal{h}}$ and \eqref{2gtc}, \eqref{XXY}.
	\end{proof}
	
	Similarly, take an action
	\begin{align}
		S=\int_{M}2\langle \Omega_1, \Omega_2\rangle_{\mathcal{g}, \mathcal{h}},
	\end{align}
	and the action is 2-gauge invariant according to  the Lemma \ref{lemma1}. Without loss of generality,  the action also gives completely trivial equations. A routine computation is given by
	\begin{align}
		\delta S &=2 \int_{M} \langle \delta \Omega_1, \Omega_2 \rangle_{\mathcal{g}, \mathcal{h}} + \langle \Omega_1, \delta \Omega_2 \rangle_{\mathcal{g}, \mathcal{h}}\nonumber\\
		&=2 \int_{M} \langle d \delta A + A \wedge^{[, ]}\delta A- \alpha(\delta B), \Omega_2 \rangle_{\mathcal{g},\mathcal{h}} + \langle \Omega_1, d \delta B + \delta A \wedge^{\vartriangleright} B + A \wedge^{\vartriangleright} \delta B \rangle_{\mathcal{g}, \mathcal{h}}\nonumber\\
		&= 2\int_{M} \langle \delta A , d \Omega_2 + A \wedge^{\vartriangleright}\Omega_2 - \Omega_1 \wedge^{\vartriangleright}B \rangle_{\mathcal{g}, \mathcal{h}} - \langle d \Omega_1 + A \wedge^{[, ]}\Omega_1 + \alpha(\Omega_2), \delta B \rangle_{\mathcal{g}, \mathcal{h}}
	\end{align}
	It is apparent from the above equation that $\delta S =0$ for all $(A, B)$ by using the 2-Bianchi Identities \eqref{2BI}.
	
	Similar to the Chern-Weil theorem in the ordinary CS gauge theory, we can develop a 2-Chern-Weil theorem for the second 2-Chern form. The proof is analogous to that in Theorem \ref{theorem1}.
	\begin{theorem}[2-Chern-Weil theorem]
		The second 2-Chern form $P(\Omega_1, \Omega_2)$ satisfies
		\begin{itemize}
			\item[1).] $P(\Omega_1, \Omega_2)$ is closed, i.e. $d P(\Omega_1, \Omega_2)=0$;
			\item[2).] $P(\Omega_1, \Omega_2)$ has topologically invariant integral, namely satisfying the 2-Chern-Weil homomorphism formula:
			\begin{align}\label{PPQ}
				P( \Omega^1_1, \Omega^1_2) - P (\Omega^0_1, \Omega^0_2)=d Q(A^1, A^0, B^1, B^0),
			\end{align}
			with 
			\begin{align}\label{Q4}
				Q(A^1, A^0, B^1, B^0)= \int_{0}^{1} P(A^1- A^0, \Omega^t_2) + P	(\Omega^t_1, B^1 - B^0) dt,
			\end{align}
			where $(A^0, B^0)$ and $(A^1, B^1)$ are two 2-connections, $(\Omega^0_1, \Omega^0_2)$ and $(\Omega^1_1, \Omega^1_2)$ the corresponding 2-curvatures,
			\begin{align}
				&A^t = A^0 + t \eta, \ \ \eta= A^1- A^0,\\
				&B^t = B^0 + t \overline{\eta},\ \ \overline{\eta}= B^1- B^0,\ \ \ \ (0 \leq t \leq 1)
			\end{align}
			the interpolation between $(A^0, B^0)$ and $(A^1, B^1)$,
			\begin{align}\label{AB}
				\Omega^t_1 = d A^t + A^t \wedge A^t - \alpha(B^t),\ \ \ 
				\Omega^t_2 = dB^t + A^t \wedge^{\vartriangleright}B^t,
			\end{align}
			the 2-curvature of this interpolation.
		\end{itemize}
	\end{theorem}
	\begin{proof}
		\begin{itemize}
			\item[1).] 	Using the 2-Bianchi Identities, we can verify directly that
			\begin{align}
				\frac{1}{2}d P(\Omega_1, \Omega_2) = &d\langle \Omega_1, \Omega_2\rangle_{\mathcal{g}, \mathcal{h}}\nonumber\\
				=&\langle d\Omega_1, \Omega_2\rangle_{\mathcal{g}, \mathcal{h}} + \langle \Omega_1, d\Omega_2\rangle_{\mathcal{g}, \mathcal{h}}\nonumber\\
				=&\langle A \wedge^{[, ]}\Omega_1 + \alpha(\Omega_2), \Omega_2\rangle_{\mathcal{g}, \mathcal{h}} + \langle \Omega_1, A \wedge^{\vartriangleright} \Omega_2 - \Omega_1 \wedge^{\vartriangleright} B\rangle_{\mathcal{g}, \mathcal{h}}\nonumber\\
				=&0,
			\end{align}
			by using \eqref{XXY} and \eqref{XY}.
			\item[2).] Differentiating both $\Omega^t_1$ and $\Omega^t_2$ with respect to $t$ gives
			\begin{align}
				\frac{d}{dt}\Omega^t_1 = d \eta + A^t \wedge^{[, ]}\eta - \alpha(\overline{\eta}),\\
				\frac{d}{dt}\Omega^t_2 = d \overline{\eta} + A^t \wedge^{\vartriangleright} \overline{\eta} + \eta \wedge^{\vartriangleright}B^t.
			\end{align}
			Then 
			\begin{align}\label{Pt}
				\frac{d}{dt}P(\Omega^t_1, \Omega^t_2)&= 2 \langle \frac{d}{dt}\Omega^t_1, \Omega^t_2 \rangle_{\mathcal{g}, \mathcal{h}} +2 \langle \Omega^t_1, \frac{d}{dt}\Omega^t_2 \rangle_{\mathcal{g}, \mathcal{h}} \nonumber\\
				&=2\langle d \eta + A^t \wedge^{[, ]}\eta - \alpha(\overline{\eta}), \Omega^t_2 \rangle_{\mathcal{g}, \mathcal{h}} +2 \langle \Omega^t_1, d \overline{\eta} + A^t \wedge^{\vartriangleright} \overline{\eta} + \eta \wedge^{\vartriangleright}B^t \rangle_{\mathcal{g}, \mathcal{h}}.
			\end{align}
			On the other hand, we have
			\begin{align}\label{Ptt}
				d P(\eta, \Omega^t_2)&=2\langle d \eta, \Omega^t_2\rangle_{\mathcal{g}, \mathcal{h}} - 2\langle \eta, d\Omega^t_2 \rangle_{\mathcal{g}, \mathcal{h}} \nonumber\\
				&=2\langle d \eta, \Omega^t_2\rangle_{\mathcal{g}, \mathcal{h}} -2\langle \eta, \Omega^t_1 \wedge^{\vartriangleright}B^t - A^t \wedge^{\vartriangleright}\Omega^t_2 \rangle_{\mathcal{g}, \mathcal{h}}\nonumber\\
				&=2\langle d \eta + A^t\wedge^{[, ]}\eta, \Omega^t_2 \rangle_{\mathcal{g}, \mathcal{h}} + 2\langle \Omega^t_1, \eta \wedge^{\vartriangleright}B^t \rangle_{\mathcal{g}, \mathcal{h}}
			\end{align}
			by using \eqref{XXY}, and
			\begin{align}\label{Pttt}
				d P(\Omega^t_1, \overline{\eta})&=2\langle d \Omega^t_1, \overline{\eta} \rangle_{\mathcal{g}, \mathcal{h}}+ 2\langle \Omega^t_1, d \overline{\eta} \rangle_{\mathcal{g}, \mathcal{h}}\nonumber\\
				&= -2\langle A^t \wedge^{[, ]}\Omega^t_1 + \alpha(\Omega^t_2), \overline{\eta} \rangle_{\mathcal{g}, \mathcal{h}} + 2\langle \Omega^t_1, d \overline{\eta} \rangle_{\mathcal{g}, \mathcal{h}}\nonumber\\
				&=2\langle \Omega^t_1, d \overline{\eta} + A^t \wedge^{\vartriangleright}\overline{\eta} \rangle_{\mathcal{g}, \mathcal{h}} -2 \langle \alpha(\overline{\eta}), \Omega^t_2 \rangle_{\mathcal{g}, \mathcal{h}},
			\end{align}
			by using \eqref{XY}. 
			
			By comparing \eqref{Pt}, \eqref{Ptt} and \eqref{Pttt}, we obtain
			\begin{align}
				\frac{d}{dt}P(\Omega^t_1, \Omega^t_2)= d P(\eta, \Omega^t_2)+ d P(\Omega^t_1, \overline{\eta}),
			\end{align}
			and integrating the both sides of  this equation of $t$ between the limits $0$ and $1$, we have
			\begin{align}
				P( \Omega^1_1, \Omega^1_2) - P (\Omega^0_1, \Omega^0_2)=d \int_{0}^{1} P(\eta, \Omega^t_2)+ P(\Omega^t_1, \overline{\eta})dt=
				d Q(A^1, A^0, B^1, B^0).
			\end{align}
		\end{itemize}
	\end{proof}
	
	This shows that $P( \Omega^1_1, \Omega^1_2)$ and $P (\Omega^0_1, \Omega^0_2)$ differ by an exact form. In other wards, their integrals over 5-d manifolds without boundary give the same results, and we call $Q(A^1, A^0, B^1, B^0)$ the secondary 2-topological class.
	
	In particular, consider $A^0=0$, $B^0=0$, $A^1=A$ and $B^1=B$ in \eqref{PPQ}, and we have
	\begin{align}
		2\langle \Omega_1, \Omega_2\rangle_{\mathcal{g}, \mathcal{h}} =d Q_{2CS}.
	\end{align}
	We call $Q_{2CS}$ the \textbf{ 2-Chern-Simons form}, having
	\begin{align}
		Q_{2CS}&= 2\int_{0}^{1}\langle A, t dB + t^2 A \wedge^{\vartriangleright}B \rangle_{\mathcal{g}, \mathcal{h}} + \langle t dA + t^2 A \wedge A - t \alpha(B), B\rangle_{\mathcal{g}, \mathcal{h}}dt\nonumber\\
		&= \langle A, d B + \frac{2}{3}A\wedge^{\vartriangleright}B\rangle_{\mathcal{g}, \mathcal{h}} + \langle dA + \frac{2}{3}A \wedge A - \alpha(B), B \rangle_{\mathcal{g}, \mathcal{h}}.
	\end{align}
	Then we obtain the 2CS action on a 4-dimensional manifold $M$
	\begin{align}
		S_{2CS}= \frac{k}{4 \pi}\int_{M} \langle A, d B + \frac{2}{3}A\wedge^{\vartriangleright}B\rangle_{\mathcal{g}, \mathcal{h}} + \langle dA + \frac{2}{3}A \wedge A - \alpha(B), B \rangle_{\mathcal{g}, \mathcal{h}}.
	\end{align}

	%
	
	For these reasons, by its analogy to the standard CS
	theory and as implied by its given name, the present model can be legitimately considered a 2CS gauge theory.

	\subsection{3CS gauge theory and  3-Chern-Weil theorem}\label{sub63}
	
	In this section, we construct a $3CS$ gauge theory based on a similar heuristic argument with the $2CS$ gauge theory. Similar arguments apply to the type $N=2$ generalized connection, and the 3CS gauge theory exists in 5-dimensional manifolds and is built in the framework of the 3-gauge theory. Besides, the 3-gauge symmetry of this theory is encoded by a Lie 2-crossed module.

	Given a 2-crossed module $W=(L, H, G; \beta, \alpha, \vartriangleright, \{, \})$,  there is an associated differential 2-crossed module $\mathcal{W}=(\mathcal{l}, \mathcal{h}, \mathcal{g}; \beta, \alpha, \vartriangleright, \{, \})$. 
	Let $(A, B, C)$ be a 3-connection on the principle 3-bundle $P(M, W)$.
	The 3CS action is defined by a type $N=2$ generalized connection $	\mathcal{A}= A + B \xi^1+ B \xi^2 + C \xi^{12}$ as follows
	\begin{align}\label{51}
		S_{3CS}=\frac{k}{4 \pi}\int \ll \mathcal{A}, d \mathcal{A} + \frac{1}{3}\mathcal{A}\wedge^{[,]}\mathcal{A}\gg.
	\end{align}
	Using \eqref{gl3cs}, \eqref{N2} and \eqref{N22} with $k^1=0$ and $k^2=-1$, a straightforward computation gives
	\begin{align}\label{3csac}
		S_{3CS}
		=&\frac{k}{4 \pi}\int \ll A + B \xi^1+ B \xi^2 + C \xi^{12}, dA -\alpha(B) + \frac{2}{3}A \wedge A + (d B - \beta(C) + \frac{2}{3}A \wedge^{\vartriangleright}B)\xi^1 \nonumber\\
		&+ (dB +\frac{2}{3}A \wedge^{\vartriangleright}B )\xi^2 + (dC + \frac{2}{3}A \wedge^{\vartriangleright}C + \frac{2}{3}B \wedge^{\{, \}}B)\xi^{12}\gg\nonumber\\
		=&\frac{k}{4 \pi}\int \langle A, dC + \frac{2}{3}A \wedge^{\vartriangleright}C + \frac{2}{3}B \wedge^{\{, \}}B\rangle_{\mathcal{g}, \mathcal{l}} + \langle dA -\alpha(B) + \frac{2}{3}A \wedge A, C\rangle_{\mathcal{g}, \mathcal{l}} \nonumber\\
		&+ \langle B, d B - \beta(C) + \frac{2}{3}A \wedge^{\vartriangleright}B\rangle_{\mathcal{h}}.
	\end{align}
	For the first part, have
	\begin{align}\label{f}
		&\langle A, dC + \frac{2}{3} A \wedge^{\vartriangleright}C + \frac{2}{3}B \wedge^{\{, \}}B \rangle_{\mathcal{g}, \mathcal{l}} \nonumber \\
		&= \langle A, dC \rangle_{\mathcal{g}, \mathcal{l}} + \frac{2}{3}\langle A \wedge^{[, ]} A, C \rangle_{\mathcal{g}, \mathcal{l}} + \frac{2}{3}\langle A, B \wedge^{\{, \}}B \rangle_{\mathcal{g}, \mathcal{l}},\nonumber\\
		&= \langle dA, C \rangle_{\mathcal{g}, \mathcal{l}} + \frac{4}{3}\langle A \wedge A, C \rangle_{\mathcal{g}, \mathcal{l}} + \frac{1}{3}\langle B, A\wedge^{\vartriangleright} B \rangle_{\mathcal{h}},\nonumber\\
		&= \langle dA + \frac{4}{3}A \wedge A, C\rangle_{\mathcal{g}, \mathcal{l}} +  \frac{1}{3}\langle B, A\wedge^{\vartriangleright} B\rangle_{\mathcal{h}}
	\end{align}
	by using \eqref{YX}, \eqref{XZ} and \eqref{XYY}.
	By substituting \eqref{f} into \eqref{3csac}, we obtain 
	\begin{align}\label{3CS}
		S_{3CS}& =\frac{k}{4 \pi}\int  \langle 2F- \alpha(B), C\rangle_{\mathcal{g}, \mathcal{l}} +\langle B, \Omega_2 \rangle_{\mathcal{h}}.
	\end{align}
	So, the 3CS theory can be described as a 5-d deformed 2BF theory in the case that $\mathcal{g}=\mathcal{l}$.
	
	Taking the variational derivative of the action $S_{3CS}$,  we obtain the equations of motion:
	\begin{align}
		&\delta A: dC + A \wedge^{\vartriangleright}C + B \wedge^{\{, \}}B=\Omega_3=0,\\
		&\delta B: dB + A \wedge^{\vartriangleright}B - \beta(C)=\Omega_2=0,\\
		&\delta C: dA + A \wedge A - \alpha(B)=\Omega_1=0.
	\end{align}
	Then  the solution of these equations gives a flat 3-connection $(A, B, C)$ on the trivial  principal 3-bundle, analogously to standard CS gauge theory.
	
	
	Similar considerations about Chern form apply to the 3CS gauge theory.
	Let us first introduce the \textbf{second 3-Chern form} defined by the 3-curvature $(\Omega_1, \Omega_2, \Omega_3)$
	\begin{align}
		P(\Omega_1, \Omega_2, \Omega_3) 
		= 2 \langle \Omega_1, \Omega_3 \rangle_{\mathcal{g}, \mathcal{l}} + \langle \Omega_2, \Omega_2 \rangle_{\mathcal{h}}.
	\end{align}
	The lemma follows will tell us the second 3-Chern form is 3-gauge invariant.
	\begin{lemma}
		$P(\Omega_1, \Omega_2, \Omega_3)= 2 \langle \Omega_1, \Omega_3 \rangle_{\mathcal{g}, \mathcal{l}} + \langle \Omega_2, \Omega_2 \rangle_{\mathcal{h}}$ is 3-gauge invariant under the general 3-gauge transformation \eqref{g3gt}.
	\end{lemma}
	\begin{proof}
		Under the general 3-gauge transformation \eqref{g3gt}, have
		\begin{align}
			P(\overline{\Omega}_1, \overline{\Omega}_2, \overline{\Omega}_3)= 2 \langle \overline{\Omega}_1, \overline{\Omega}_3\rangle_{\mathcal{g}, \mathcal{l}} + \langle \overline{\Omega}_2, \overline{\Omega}_2 \rangle_{\mathcal{h}}. 
		\end{align}
		For the first section, have
		\begin{align}
			2 \langle \overline{\Omega}_1, \overline{\Omega}_3\rangle_{\mathcal{g}, \mathcal{l}}
			&=2 \langle g^{-1}\Omega_1g, g^{-1}\vartriangleright \Omega_3 - ( g^{-1}\vartriangleright \Omega_2 + (g^{-1}\Omega_1 g)\wedge^{\vartriangleright}\phi) \wedge^{\{, \}} \phi \nonumber\\
			&\ \ \ + \phi \wedge^{\{, \}} (g^{-1}\vartriangleright \Omega_2) - (g^{-1}\Omega_1 g)\wedge^{\vartriangleright} \psi \rangle_{\mathcal{g}, \mathcal{l}}\nonumber\\
			&= 2\langle \Omega_1, \Omega_3 \rangle_{\mathcal{g}, \mathcal{l}} - 2 \langle g^{-1} \Omega_1 g, ( g^{-1}\vartriangleright \Omega_2 + (g^{-1}\Omega_1 g)\wedge^{\vartriangleright}\phi) \wedge^{\{, \}} \phi  \rangle_{\mathcal{g}, \mathcal{l}}\nonumber\\
			&\ \ \ + 2\langle g^{-1} \Omega_1 g, \phi \wedge^{\{, \}} (g^{-1}\vartriangleright \Omega_2)\rangle_{\mathcal{g}, \mathcal{l}}
			- 2\langle g^{-1} \Omega_1 g, (g^{-1}\Omega_1 g)\wedge^{\vartriangleright} \psi \rangle_{\mathcal{g}, \mathcal{l}}\nonumber\\
			&= 2\langle \Omega_1, \Omega_3 \rangle_{\mathcal{g}, \mathcal{l}} + \langle\phi, (g^{-1}\vartriangleright \Omega_1)\wedge^{\vartriangleright}(g^{-1}\vartriangleright \Omega_2 + (g^{-1}\Omega_1g)\wedge^{\vartriangleright}\phi)\rangle_{\mathcal{h}} \nonumber\\
			&\ \ \ - \langle g^{-1} \vartriangleright\Omega_2, (g^{-1}\Omega_1g)\wedge^{\vartriangleright}\phi \rangle_{\mathcal{h}} -2\langle \Omega_1, \Omega_1\wedge^{\vartriangleright}\psi \rangle_{\mathcal{g}, \mathcal{l}}\nonumber\\
			&= 2\langle \Omega_1, \Omega_3 \rangle_{\mathcal{g}, \mathcal{l}} + \langle \phi, (g^{-1}\Omega_1g)\wedge^{\vartriangleright}(g^{-1}\vartriangleright \Omega_2)\rangle_{\mathcal{h}} - \langle g^{-1}\vartriangleright \Omega_2, (g^{-1}\Omega_1g)\wedge^{\vartriangleright}\phi \rangle_{\mathcal{h}},
		\end{align}
		and for the second section, 
		\begin{align}
			\langle \overline{\Omega}_2, \overline{\Omega}_2 \rangle_{\mathcal{h}} &= \langle g^{-1}\vartriangleright \Omega_2 + (g^{-1}\Omega_1g)\wedge^{\vartriangleright}\phi,  g^{-1}\vartriangleright \Omega_2 + (g^{-1}\Omega_1g)\wedge^{\vartriangleright}\phi
			\rangle_{\mathcal{h}}\nonumber\\
			&=\langle \Omega_2, \Omega_2 \rangle_{\mathcal{h}} + \langle g^{-1}\vartriangleright \Omega_2, (g^{-1}\Omega_1g)\wedge^{\vartriangleright}\phi \rangle_{\mathcal{h}}+ \langle (g^{-1}\Omega_1g)\wedge^{\vartriangleright}\phi, g^{-1}\vartriangleright \Omega_2 \rangle_{\mathcal{h}}\nonumber\\
			&\ \ \ +\langle \Omega_1\wedge^{\vartriangleright}\phi, \Omega_1 \wedge^{\vartriangleright}\phi \rangle_{\mathcal{h}}\nonumber\\
			&= \langle \Omega_2, \Omega_2 \rangle_{\mathcal{h}} - \langle \phi, (g^{-1} \Omega_1g)\wedge^{\vartriangleright}(g^{-1}\vartriangleright \Omega_2)\rangle_{\mathcal{h}} + \langle g^{-1}\vartriangleright \Omega_2, (g^{-1}\Omega_1g)\wedge^{\vartriangleright}\phi \rangle_{\mathcal{h}}
		\end{align}
		by using \eqref{YX}, \eqref{XXY} and \eqref{XYY}.Then
		\begin{align}
			P(\overline{\Omega}_1, \overline{\Omega}_2, \overline{\Omega}_3)= P(\Omega_1, \Omega_2, \Omega_3).
		\end{align}
	\end{proof}
	
	The same manipulations work apply to the second 3-Chern form, having 
	\begin{align}
		S= \int_{M} 2 \langle \Omega_1, \Omega_3 \rangle_{\mathcal{g}, \mathcal{l}} + \langle \Omega_2, \Omega_2 \rangle_{\mathcal{h}},
	\end{align}
	and varying this action with respect to $A$, $B$ and $C$ gives
	\begin{align}
		\delta S = & 2\int_{M} \langle \delta \Omega_1, \Omega_3 \rangle_{\mathcal{g},\mathcal{l}} + \langle \Omega_1, \delta \Omega_3 \rangle_{\mathcal{g},\mathcal{l}} + \langle \delta \Omega_2, \Omega_2 \rangle_{ \mathcal{h}}\nonumber\\
		=& 2\int_{M}\langle d \delta A + A \wedge^{[, ]}\delta A - \alpha(\delta B), \Omega_3 \rangle_{\mathcal{g}, \mathcal{l}} + \langle \Omega_1, d \delta C + \delta A \wedge^{\vartriangleright}C + A \wedge^{\vartriangleright}\delta C + \delta B \wedge^{\{, \}}B \nonumber\\
		&+ B \wedge^{\{, \}}\delta B \rangle_{\mathcal{g}, \mathcal{l}} + \langle d \delta B + \delta A \wedge^{\vartriangleright}B + A \wedge^{\vartriangleright}\delta B - \beta(\delta C), \Omega_2 \rangle_{\mathcal{h}}\nonumber\\
		=&2\int_{M} \langle \delta A, d \Omega_3 \rangle_{\mathcal{g}, \mathcal{l}} + \langle \delta A, A \wedge^{\vartriangleright}\Omega_3 \rangle_{\mathcal{g}, \mathcal{l}} + \langle \beta(\Omega_3), \delta B \rangle_{\mathcal{h}} -\langle d\Omega_1, \delta C \rangle_{\mathcal{g}, \mathcal{l}} -\langle \delta A, \Omega_1 \wedge^{\vartriangleright}C \rangle_{\mathcal{g}, \mathcal{l}}\nonumber\\
		&- \langle A\wedge^{[, ]}\Omega_1, \delta C \rangle_{\mathcal{g}, \mathcal{l}} + \frac{1}{2}\langle B, \Omega_1 \wedge^{\vartriangleright}\delta B \rangle_{\mathcal{h}} + \frac{1}{2}\langle \delta B, \Omega_1 \wedge^{\vartriangleright} B \rangle_{\mathcal{h}} - \langle \delta B, d \Omega_2 \rangle_{\mathcal{h}} \nonumber\\
		&+ \langle \delta A\wedge^{\vartriangleright}\Omega_2, B \rangle_{\mathcal{h}} + \langle A \wedge^{\vartriangleright}\Omega_2, \delta B \rangle_{\mathcal{h}} - \langle \alpha(\Omega_2), \delta C \rangle_{\mathcal{g}, \mathcal{l}},\label{eS}
	\end{align}
	by using \eqref{XZ}--\eqref{XYY} and \eqref{YX}.
	Note that 
	\begin{align}
		\langle B, \Omega_1 \wedge^{\vartriangleright}\delta B \rangle_{\mathcal{h}} &= \langle \delta B, \Omega_1 \wedge^{\vartriangleright}B \rangle_{\mathcal{h}},\label{e1}\\
		\langle \delta A\wedge^{\vartriangleright}\Omega_2, B \rangle_{\mathcal{h}} &= - 2\langle \delta A, \Omega_2 \wedge^{\{, \}}B \rangle_{\mathcal{g}, \mathcal{l}}\label{e2},\\
		\langle B, \delta A \wedge^{\vartriangleright}\Omega_2 \rangle_{\mathcal{h}} &= - \langle \Omega_2, \delta A \wedge^{\vartriangleright}B \rangle_{\mathcal{h}} =2 \langle \delta A, B \wedge^{\{, \}}\Omega_2 \rangle_{\mathcal{g}, \mathcal{l}},\label{e3}
	\end{align}
	by using \eqref{YX} and \eqref{XYY}.
	Substitute \eqref{e1}--\eqref{e3} to \eqref{eS} to get
	\begin{align}
		\delta S = &2\int_{M}\langle \delta A, 	d \Omega_3 + A \wedge^{\vartriangleright}\Omega_3 - \Omega_1\wedge^{\vartriangleright}C - B\wedge^{\{, \}}\Omega_2 - \Omega_2 \wedge^{\{, \}}B \rangle_{\mathcal{g}, \mathcal{l}} \nonumber\\
		&- \langle \delta B, d \Omega_2 + A \wedge^{\vartriangleright}\Omega_2 - \Omega_1 \wedge^{\vartriangleright}B + \beta(\Omega_3)\rangle_{\mathcal{h}} - \langle d \Omega_1 + A \wedge^{[, ]} \Omega_1 + \alpha(\Omega_2), \delta C \rangle_{\mathcal{g}, \mathcal{l}}.
	\end{align}
	It is apparent from the above equation that $\delta S=0$ for all $(A, B, C)$ by using the 3-Bianchi Identities \eqref{3bi}. Then the action also gives completely trivial equations.
	
	\begin{theorem}[3-Chern-Weil theorem]
		The second 3-Chern form $P(\Omega_1, \Omega_2, \Omega_3)$ satisfies
		\begin{itemize}
			\item[1).] $P(\Omega_1, \Omega_2, \Omega_3)$ is closed, i.e. $d P(\Omega_1, \Omega_2, \Omega_3)=0$;
			\item[2).] $P(\Omega_1, \Omega_2, \Omega_3)$ has topologically invariant integral. Namely, it satisfies the 3-Chern-Weil homomorphism formula:
			\begin{align}\label{PPQQ}
				P( \Omega^1_1, \Omega^1_2, \Omega^1_3) - P (\Omega^0_1, \Omega^0_2, \Omega^0_3)=d Q(A^1, A^0, B^1, B^0, C^1, C^0),
			\end{align}
			with 
			\begin{align}\label{Q43}
				Q(A^1, A^0, B^1, B^0, C^1, C^0)= 2 \int_{0}^{1} \langle A^1- A^0, \Omega^t_3\rangle_{\mathcal{g}, \mathcal{l}} + \langle \Omega^t_1, C^1 - C^0\rangle_{\mathcal{g}, \mathcal{l}} + \langle B^1 - B^0, \Omega^t_2\rangle_{\mathcal{h}}dt,
			\end{align}
			where $(A^0, B^0, C^0)$ and $(A^1, B^1, C^1)$ are two 3-connections, $(\Omega^0_1, \Omega^0_2, \Omega^0_3)$ and $(\Omega^1_1, \Omega^1_2, \Omega^1_3)$ the corresponding 3-curvatures,
			\begin{align}
				&A^t = A^0 + t \eta, \ \ \eta= A^1- A^0,\\
				&B^t = B^0 + t \overline{\eta},\ \ \overline{\eta}= B^1- B^0,\\
				&C^t = C^0 + t \tilde{\eta}, \ \ \tilde{\eta}= C^1- C^0, \ \ \ (0 \leq t \leq 1)
			\end{align}
			the interpolation between $(A^0, B^0, C^0)$ and $(A^1, B^1, C^1)$,
			\begin{align}\label{ABC}
				\Omega^t_1 &= d A^t + A^t \wedge A^t - \alpha(B^t),\\
				\Omega^t_2 &= dB^t + A^t \wedge^{\vartriangleright}B^t-\beta(C^t),\nonumber\\
				\Omega^t_3&= dC^t + A^t \wedge^{\vartriangleright}C^t + B^t \wedge^{\{, \}} B^t,
			\end{align}
			the 3-curvature of this interpolation.
		\end{itemize}
	\end{theorem}
	\begin{proof}
		\begin{itemize}
			\item[1).]
			Using the 3-Bianchi Identities \eqref{3bi}, have
			\begin{align}
				& \frac{1}{2}d P(\Omega_1, \Omega_2, \Omega_3)\nonumber\\
				=&\langle d \Omega_1, \Omega_3\rangle_{\mathcal{g}, \mathcal{l}} + \langle \Omega_1 d \Omega_3 \rangle_{\mathcal{g}, \mathcal{l}} + \langle d\Omega_2, \Omega_2\rangle_{\mathcal{h}}\nonumber\\
				=&\langle A \wedge^{[, ]}\Omega_1 + \alpha(\Omega_2), \Omega_3 
				\rangle_{\mathcal{g}, \mathcal{l}} + \langle \Omega_1, A \wedge^{\vartriangleright} \Omega_3 - \Omega_1\wedge^{\vartriangleright}C - B \wedge^{\{, \}}\Omega_2 - \Omega_2 \wedge^{\{, \}}B \rangle_{\mathcal{g}, \mathcal{l}}\nonumber\\
				& + \langle A \wedge^{\vartriangleright}\Omega_2 - \Omega_1 \wedge^{\vartriangleright}B + \beta(\Omega_3), \Omega_2 \rangle_{\mathcal{h}}\nonumber\\
				=&-\langle \Omega_1, A \wedge^{\vartriangleright}\Omega_3\rangle_{\mathcal{g}, \mathcal{l}} - \langle \beta(\Omega_3), \Omega_2 \rangle_{\mathcal{h}} +  \langle \Omega_1, A \wedge^{\vartriangleright} \Omega_3 - \Omega_1\wedge^{\vartriangleright}C \rangle_{\mathcal{g}, \mathcal{l}} - \frac{1}{2}\langle \Omega_2, \Omega_1 \wedge^{\vartriangleright}B\rangle_{\mathcal{h}} \nonumber\\
				&-\frac{1}{2}\langle B, \Omega_1\wedge^{\vartriangleright}\Omega_2 \rangle_{\mathcal{h}} + \langle A \wedge^{\vartriangleright}\Omega_2, \Omega_2 \rangle_{\mathcal{h}} + \langle \Omega_2, \Omega_1 \wedge^{\vartriangleright}B \rangle_{\mathcal{h}} + \langle \beta(\Omega_3), \Omega_2 \rangle_{\mathcal{h}}\nonumber\\
				=0
			\end{align}
			by using \eqref{YX}--\eqref{XYY}.
			\item[2).] Differentiating $\Omega^t_1$, $\Omega^t_2$ and $\Omega^t_3$ with respect to $t$ gives 
			\begin{align}
				\frac{d}{dt}\Omega^t_1 = d \eta + A^t \wedge^{[, ]}\eta - \alpha(\overline{\eta}),\\
				\frac{d}{dt}\Omega^t_2 = d \overline{\eta} + A^t \wedge^{\vartriangleright} \overline{\eta} + \eta \wedge^{\vartriangleright}B^t- \beta(\tilde{\eta}),\\
				\frac{d}{dt}\Omega^t_3 = d \tilde{\eta} + A^t \wedge^{\vartriangleright}\tilde{\eta} + \eta \wedge^{\vartriangleright}C^t + B^t \wedge^{\{, \}}\overline{\eta} + \overline{\eta} \wedge^{\{, \}}B^t.
			\end{align}
			Then, have
			\begin{align}
				\frac{d}{dt}P(\Omega^t_1, \Omega^t_2, \Omega^t_3)=&2 \langle \frac{d}{dt}\Omega^t_1, \Omega^t_3\rangle_{\mathcal{g}, \mathcal{l}} + 2\langle \Omega^t_1, 	\frac{d}{dt}\Omega^t_3\rangle_{\mathcal{g}, \mathcal{l}} +2\langle \frac{d}{dt} \Omega^t_2,  \Omega^t_2\rangle_{\mathcal{h}}\nonumber\\
				=&2 \langle d \eta + A^t \wedge^{[, ]}\eta - \alpha(\overline{\eta}), \Omega^t_3\rangle_{\mathcal{g}, \mathcal{l}} + 2\langle \Omega^t_1, d \tilde{\eta} + A^t \wedge^{\vartriangleright}\tilde{\eta} + \eta \wedge^{\vartriangleright}C^t \nonumber\\
				&+ B^t \wedge^{\{, \}}\overline{\eta} + \overline{\eta} \wedge^{\{, \}}B^t\rangle_{\mathcal{g}, \mathcal{l}} +2\langle d \overline{\eta} + A^t \wedge^{\vartriangleright} \overline{\eta} + \eta \wedge^{\vartriangleright}B^t- \beta(\tilde{\eta}), \Omega^t_2\rangle_{\mathcal{h}}.
			\end{align}
			On the other hand, we have
			\begin{align}
				2 d \langle \eta, \Omega^t_3 \rangle_{\mathcal{g}, \mathcal{l}}&= 2\langle d \eta, \Omega^t_3 \rangle_{\mathcal{g}, \mathcal{l}} - 2\langle \eta, d \Omega^t_3 \rangle_{\mathcal{g}, \mathcal{l}}\nonumber\\
				&=2\langle d \eta, \Omega^t_3 \rangle_{\mathcal{g}, \mathcal{l}} - 2\langle \eta, \Omega^t_1 \wedge^{\vartriangleright}C^t + B^t \wedge^{\{, \}}\Omega^t_2 + \Omega^t_2 \wedge^{\{, \}}B^t - A^t \wedge^{\vartriangleright} \Omega^t_3 \rangle_{\mathcal{g}, \mathcal{l}}\nonumber\\
				&=2\langle d \eta + A^t \wedge^{[, ]}\eta, \Omega^t_3 \rangle_{\mathcal{g}, \mathcal{l}} + 2\langle \Omega^t_1, \eta \wedge^{\vartriangleright} C^t \rangle_{\mathcal{g}, \mathcal{l}} + 2\langle \Omega^t_2, \eta \wedge^{\vartriangleright} B^t \rangle_{\mathcal{h}},\label{6}\\
				2 d \langle \Omega^t_1, \tilde{\eta} \rangle_{\mathcal{g}, \mathcal{l}}&=2 \langle d\Omega^t_1, \tilde{\eta} \rangle_{\mathcal{g}, \mathcal{l}} + 2 \langle \Omega^t_1, d\tilde{\eta} \rangle_{\mathcal{g}, \mathcal{l}}\nonumber\\
				&=-2\langle A^t \wedge^{[, ]}\Omega^t_1 + \alpha(\Omega^t_2), \tilde{\eta} \rangle_{\mathcal{g}, \mathcal{l}} + 2 \langle \Omega^t_1, d\tilde{\eta} \rangle_{\mathcal{g}, \mathcal{l}}\nonumber\\
				&=2\langle \Omega^t_1, A^t \wedge^{\vartriangleright}\tilde{\eta} \rangle_{\mathcal{g}, \mathcal{l}} - 2 \langle \beta(\tilde{\eta}), \Omega^t_2 \rangle_{\mathcal{h}} + 2 \langle \Omega^t_1, d\tilde{\eta} \rangle_{\mathcal{g}, \mathcal{l}},\label{66}\\
				2 d\langle \overline{\eta}, \Omega^t_2 \rangle_{\mathcal{h}}&=2 \langle d\overline{\eta}, \Omega^t_2 \rangle_{\mathcal{h}} + 2 \langle \overline{\eta}, d\Omega^t_2 \rangle_{\mathcal{h}}\nonumber\\
				&=2 \langle d\overline{\eta}, \Omega^t_2 \rangle_{\mathcal{h}} +2 \langle \overline{\eta}, \Omega^t_1 \wedge^{\vartriangleright}B^t - A^t \wedge^{\vartriangleright}\Omega^t_2 - \beta(\Omega^t_3)\rangle_{\mathcal{h}}\nonumber\\
				&=2\langle d \overline{\eta} + A^t \wedge^{\vartriangleright}\overline{\eta}, \Omega^t_2 \rangle_{\mathcal{h}} + 2\langle \Omega^t_1, B^t \wedge^{\{, \}}\overline{\eta} + \overline{\eta} \wedge^{\{, \}}B^t \rangle_{\mathcal{h}} - 2\langle \alpha(\overline{\eta}), \Omega^t_3 \rangle_{\mathcal{g}, \mathcal{l}},\label{666}
			\end{align}
			by using \eqref{YX}--\eqref{XYY} and the 3-Bianchi Identities \eqref{3bi}. Comparing \eqref{6}, \eqref{66} and \eqref{666}, we obtain
			\begin{align}
				\frac{d}{dt}P(\Omega^t_1, \Omega^t_2, \Omega^t_3)=2 d \langle \eta, \Omega^t_3 \rangle_{\mathcal{g}, \mathcal{l}} + 2 d \langle \Omega^t_1, \tilde{\eta} \rangle_{\mathcal{g}, \mathcal{l}} +	2 d\langle \overline{\eta}, \Omega^t_2 \rangle_{\mathcal{h}},
			\end{align}
			and integrating the both sides of the this equation of $t$ between the limits $0$ and $1$, we have
			\begin{align}
				P( \Omega^1_1, \Omega^1_2, \Omega^1_3) - P (\Omega^0_1, \Omega^0_2, \Omega^0_3)&=2d \int_{0}^{1}  \langle \eta, \Omega^t_3 \rangle_{\mathcal{g}, \mathcal{l}} +  \langle \Omega^t_1, \tilde{\eta} \rangle_{\mathcal{g}, \mathcal{l}} +\langle \overline{\eta}, \Omega^t_2 \rangle_{\mathcal{h}}\nonumber\\
				&=d Q(A^1, A^0, B^1, B^0, C^1, C^0).
			\end{align}
			There is an alternative proof of $2)$ given in Appendix \ref{23cw}.
		\end{itemize}
	\end{proof}
	
	This shows that $P( \Omega^1_1, \Omega^1_2, \Omega^1_3)$ and $P( \Omega^0_1, \Omega^0_2, \Omega^0_3)$ differ by an exact form. Namely, their integrals over 6-d manifolds without boundary give the same results, and we call $Q(A^1, A^0, B^1, B^0, C^1, C^0)$ the secondary 3-topological class.
	
	In particular, consider $A^0=0$, $B^0=0$, $C^0=0$, $A^1=A$, $B^1=B$,and $C^1 = C$ in \eqref{PPQQ}, 
	and we have
	\begin{align}
		2 \langle \Omega_1, \Omega_3 \rangle_{\mathcal{g}, \mathcal{l}} + \langle \Omega_2, \Omega_2 \rangle_{\mathcal{h}}= d Q_{3CS},
	\end{align}
	with 
	\begin{align}
		Q_{3CS}=&2 \int_{0}^{1}\langle A, t dC+t^2 A \wedge^{\vartriangleright}C + t^2 B \wedge^{\{, \}}B \rangle_{\mathcal{g}, \mathcal{l}} + \langle t dA + t^2 A \wedge A - t \alpha(B), C \rangle_{\mathcal{g}, \mathcal{l}}\nonumber\\
		& + \langle B, t dB + t^2 A \wedge^{\vartriangleright}B - t \beta(C)\rangle_{\mathcal{h}}\nonumber\\
		=& \langle A, dC + \frac{2}{3}A \wedge^{\vartriangleright}C + \frac{2}{3}B \wedge^{\{, \}}B\rangle_{\mathcal{g}, \mathcal{l}} + \langle dA -\alpha(B) + \frac{2}{3}A \wedge A, C\rangle_{\mathcal{g}, \mathcal{l}} \nonumber\\
		&+ \langle B, d B - \beta(C) + \frac{2}{3}A \wedge^{\vartriangleright}B\rangle_{\mathcal{h}}.
	\end{align}
	We call $Q_{3CS}$ the \textbf{3-Chern-Simons form}.
	Then we obtain the 3CS action on the 5-dimensional manifold $M$
	\begin{align}
		S_{3CS}=&\frac{k}{4 \pi}\int_{M} \langle A, dC + \frac{2}{3}A \wedge^{\vartriangleright}C + \frac{2}{3}B \wedge^{\{, \}}B\rangle_{\mathcal{g}, \mathcal{l}} + \langle dA -\alpha(B) + \frac{2}{3}A \wedge A, C\rangle_{\mathcal{g}, \mathcal{l}} \nonumber\\
		&+ \langle B, d B - \beta(C) + \frac{2}{3}A \wedge^{\vartriangleright}B\rangle_{\mathcal{h}}.
	\end{align}

	%

	For these reasons, by its analogy to the standard CS
	theory and as implied by its given name, the present model can be legitimately considered a 3CS gauge theory.
	
	\section{Conclusion and outlook}
	
	In this article, we constructed the generalized differential forms valued in the differential (2-)crossed modules by using the GDC, and developed the generalized connections which consist of the higher connections. Based on the ordinary CS gauge theory, we established the 2CS and 3CS gauge theories. Finally, we generalized the second Chern form to second 2-Chern form and 3-Chern form, which satisfy the corresponding 2-Chern-Weil theorem and 3-Chern-Weil theorem, respectively.
	
	Likewise, the higher CS gauge theories may be studied using a certain kind of functional integral
	quantization which exists in the ordinary CS gauge theory. We leave the argument for future work.

	\section*{Acknowledgment}
	
	This work is supported by the National Natural Science Foundation of China (Nos.11871350, NSFC no. 11971322).
	
	The authors would like to thank the anonymous referee and editor for their valuable comments and suggestions which helped us improve the paper.
	
	\appendix

	\section{Lie  (2-)crossed modules and  differential (2-)crossed modules} \label{2cm}
	In this appendix, we collect a number of basic definitions and relations in order to define our terminology and notation and for reference throughout in
	the text. In order not to introduce additional symbols, we use some same symbols in the higher algebras. See  \cite{Radenkovic:2019qme, Martins:2010ry, TRMV1, YHMNRY, Mutlu1998FREENESSCF} for more details. The definitions and properties of the ordinary differential forms with values in  differential crossed module and 2-crossed module have been introduced in our previous paper \cite{SDH} and  related work \cite{doi:10.1063/1.4870640}.

	\textbf{Lie pre-crossed modules.}
	A Lie pre-crossed module $\left(H,G;\alpha,\vartriangleright \right)$ is given by a Lie group map $\alpha: H  \longrightarrow G$ together with a smooth left action $\vartriangleright$ of G on H by automorphisms such that:
	\begin{equation}
		\alpha \left(g \vartriangleright h\right) = g \alpha \left(h\right) g^{-1},\label{equi}
	\end{equation}
	for each $g \in G $ and $h \in H$. The Peiffer commutators in a pre-crossed module are defined as $\lbrack\lbrack \cdot , \cdot \rbrack\rbrack : H \times H \longrightarrow H $ by
	\begin{equation}
		\lbrack\lbrack h , h' \rbrack\rbrack = h h' h^{-1}\left(\alpha \left(h\right) \vartriangleright h'^{-1}\right),
	\end{equation}
	for any $h, h' \in H$.
	
	\textbf{Lie crossed modules (or a strict Lie 2-groups).}
	A Lie pre-crossed module $\left(H,G;\alpha,\vartriangleright \right)$ is said to be  a crossed module, if all of its Peiffer commutators are trivial, which is to say that
	\begin{equation}
		\alpha \left(h\right) \vartriangleright h' = h h' h^{-1} ,\label{peiffer}
	\end{equation}
	for each $h, h' \in H$. Note that the map $\left(h_1 , h_2\right) \in H \times H \longrightarrow \lbrack\lbrack h , h' \rbrack\rbrack \in H $, called the Peiffer pairing, is G-equivariant 
	\begin{equation}
		g \vartriangleright \lbrack\lbrack h_1 , h_2 \rbrack\rbrack = \lbrack\lbrack g \vartriangleright h_1 , g \vartriangleright h_2 \rbrack\rbrack ,
	\end{equation}
	for each $ h_1 , h_2 \in H$ and $g \in G$. Moreover $\lbrack\lbrack h_1 , h_2 \rbrack\rbrack =1_H $ if either $h_1$ or $h_2$ is $1_H$.
	And	\eqref{equi} and \eqref{peiffer} are called equivariance and Peiffer	properties of the crossed module, respectively.

	\textbf{Differential pre-crossed modules.} 
	A differential pre-crossed module $(\mathcal{h}, \mathcal{g}; \alpha, \vartriangleright)$ is given by a Lie algebra map $\alpha : \mathcal{h} \longrightarrow \mathcal{g}$ together with a left action $\vartriangleright$ of $\mathcal{g}$ on $\mathcal{h}$ by derivations such that:
	\begin{equation}\label{XxY}
		\alpha(X\vartriangleright Y ) =\left[ X, \alpha(Y) \right], 
	\end{equation}
	for each $X \in \mathcal{g}$ and $ Y \in \mathcal{h}$.
	The Peiffer commutators in a differential pre-crossed module are defined as $\lbrack\lbrack , \rbrack\rbrack : \mathcal{h} \times \mathcal{h} \longrightarrow \mathcal{h} $ by
	\begin{equation}\label{jkh}
		\lbrack\lbrack Y , Y'\rbrack\rbrack = \left[Y,Y'\right] - \alpha(Y)\vartriangleright Y',
	\end{equation}
	for each $Y, Y' \in \mathcal{h}$.
	
	\textbf{Differential crossed modules (or a strict Lie 2-algebras).}
	A differential pre-crossed module $(\mathcal{h}, \mathcal{g}; \alpha, \vartriangleright)$ is said to be  a differential crossed module, if all of its Peiffer commutators vanish, which is to say that: 
	\begin{equation}\label{YyY'}
		\alpha(Y)\vartriangleright Y'=\left[Y,Y'\right],
	\end{equation}
	for each $Y,Y'\in \mathcal{h}$. Note that the map $(Y_1,Y_2) \in \mathcal{h} \times \mathcal{h} \longrightarrow 	\lbrack\lbrack Y_1 , Y_2\rbrack\rbrack  \in \mathcal{h}$, called the Peiffer pairing, is $\mathcal{g}$-equivariant:
	\begin{equation}
		X \vartriangleright \lbrack\lbrack Y_1 , Y_2\rbrack\rbrack = \lbrack\lbrack X \vartriangleright Y_1 , Y_2 \rbrack\rbrack +\lbrack\lbrack  Y_1 , X \vartriangleright Y_2 \rbrack\rbrack ,
	\end{equation}
	for each $X \in \mathcal{g}$, and $Y_1,Y_2 \in \mathcal{h}$. Moreover the map $(X,Y)\in \mathcal{g} \times \mathcal{h} \longrightarrow X \vartriangleright Y \in \mathcal{h}$ is necessarily bilinear, and we have
	\begin{equation}\label{XY_1Y_2}
		X \vartriangleright \left[Y_1 ,  Y_2\right]  = \left[ X \vartriangleright Y_1 , Y_2\right] +\left[  Y_1 , X \vartriangleright Y_2\right],
	\end{equation}
	for each $X \in \mathcal{g}$ and $Y_1,Y_2 \in \mathcal{h}$, and 
	\begin{equation}\label{X_1X_2Y}
		\left[X_1, X_2\right] \vartriangleright Y =X_1 \vartriangleright (X_2 \vartriangleright Y) - X_2 \vartriangleright(X_1 \vartriangleright Y),
	\end{equation}
	for each $X_1,X_2 \in \mathcal{g}$, and $Y\in \mathcal{h}$.
	\eqref{XxY} and \eqref{YyY'} are called equivariance and Peiffer properties, respectively, in analogy to the Lie crossed module case. 
	
	\textbf{Lie 2-crossed modules.}
	A Lie 2-crossed module $(L, H, G;\beta, \alpha, \vartriangleright, \left\{,\right\})$ is given by a complex of Lie groups:
	$$L \stackrel{\beta}{\longrightarrow}H \stackrel{\alpha}{\longrightarrow}G$$
	together with smooth left action $\vartriangleright$ by automorphisms of $G$ on $L$ and $H$ (and on $G$ by conjugation), i.e.
	\begin{equation}
		g \vartriangleright (e_1 e_2)=(g \vartriangleright e_1)(g \vartriangleright e_2), \ \ \ (g_1 g_2)\vartriangleright e = g_1 \vartriangleright (g_2 \vartriangleright e),
	\end{equation}
	for any $g, g_1, g_2\in G, e, e_1, e_2\in H $ or $L$, and a $G$-equivariant smooth function $ \left\{, \right\} :H \times H \longrightarrow L $, the Peiffer lifting, such that
	\begin{equation}
		g \vartriangleright \left\{ h_1, h_2 \right\} = \left\{g \vartriangleright h_1, g \vartriangleright h_2\right\},
	\end{equation}
	for any $g \in G$ and $h_1,h_2\in H$. They satisfy:
	\begin{enumerate}
		\item $L \stackrel{\beta}{\longrightarrow}H \stackrel{\alpha}{\longrightarrow}G$ is a complex of $G$-modules (in other words $\beta$ and $\alpha$ are $G$-equivariant and $\alpha \circ \beta$ maps $L$ to $1_G$, the identity of G);
		\setlength{\itemsep}{4pt}
		\item $\beta \left\{h_1,h_2\right\} =\lbrack\lbrack h_1 , h_2 \rbrack\rbrack $, for each $h_1, h_2\in H$;
		
		\setlength{\itemsep}{4pt}
		\item $ \left[l_1, l_2\right]= \left\{ \beta (l_1), \beta (l_2)\right\}$, for each $l_1,l_2 \in L$, and here $\left[l_1,l_2\right]=l_1 l_2 l_1 ^{-1} l_2 ^{-1}$;
		
		\setlength{\itemsep}{4pt}
		\item $\left\{ h_1 h_2, h_3 \right\}= \left\{h_1, h_2 h_3 h_2 ^{-1}\right\} \alpha (h_1) \vartriangleright \left\{h_2,h_3\right\}$, for each $h_1,h_2, h_3 \in H$;
		
		\setlength{\itemsep}{4pt}
		\item $ \left\{h_1,h_2 h_3\right\}  =  \left\{ h_1 , h_2  \right\} \left\{ h_1 , h_3  \right\} \left\{ \lbrack\lbrack h_1, h_3 \rbrack\rbrack  ^{-1}, \alpha (h_1) \vartriangleright h_2 \right\}$, for each $h_1,h_2, h_3 \in H$;
		
		\setlength{\itemsep}{4pt}
		\item $\left\{ \beta (l) , h\right\} \left\{h, \beta (l)\right\} = l (\alpha (h)\vartriangleright l^{-1})$, for each $h \in H $ and $ l \in L$. 
	\end{enumerate}
	There is a left action of $H$ on $L$ by automorphisms $\vartriangleright' $ which is defined by
	\begin{equation}
		h \vartriangleright' l = l \left\{ \beta (l)^{-1}, h \right\},
	\end{equation}
	for each $l \in L$ and $h \in H$. This together with the homomorphism $ \beta : L \longrightarrow H $ defines a crossed module. In particular, for any $h\in H$,
	$$h \vartriangleright' 1_L = \left\{ 1_H, h\right\} =\left\{h, 1_H\right\}=1_L,$$
	where $1_H$ and $1_L$ are the identity of $H$ and $L$, respectively.

	\textbf{Differential 2-crossed modules.}
	A differential 2-crossed module $(\mathcal{l},\mathcal{h}, \mathcal{g}; \beta, \alpha,\vartriangleright, \left\{ , \right\})$ is given by a complex of Lie algebras:
	$$\mathcal{l}\stackrel{\beta}{\longrightarrow} \mathcal{h} \stackrel{\alpha}{\longrightarrow} \mathcal{g},$$
	together with left action $\vartriangleright $ by derivations of $\mathcal{g}$ on $\mathcal{l}, \mathcal{h}, \mathcal{g}$ (on the latter by the adjoint representation), and a $\mathcal{g}$-equivariant bilinear map $\left\{ , \right\}:\mathcal{h} \times \mathcal{h} \longrightarrow \mathcal{l}$, the Peiffer lifting, such that
	\begin{equation}\label{12}
		X\vartriangleright \left\{ Y_1,Y_2\right\} = \left\{X\vartriangleright Y_1,Y_2\right\} + \left\{Y_1, X\vartriangleright Y_2\right\},
	\end{equation}
	for each $X \in \mathcal{g}$ and $Y_1,Y_2 \in \mathcal{h}$. They satisfy:
	\begin{enumerate}
		\setlength{\itemsep}{4pt}
		\item 	$\mathcal{l}\stackrel{\beta}{\longrightarrow} \mathcal{h} \stackrel{\alpha}{\longrightarrow} \mathcal{g}$ is a complex of $\mathcal{g}$-modules, and they satisfy $\alpha \circ \beta =0 $;
		
		\setlength{\itemsep}{4pt}
		\item $\beta \left\{Y_1,Y_2\right\} =\lbrack\lbrack Y_1 , Y_2\rbrack\rbrack $, for each $Y_1, Y_2\in \mathcal{h}$, where $\lbrack\lbrack Y_1 , Y_2\rbrack\rbrack=\left[ Y_1, Y_2\right] - \alpha (Y_1) \vartriangleright Y_2$;
		
		\setlength{\itemsep}{4pt}
		\item $ \left[Z_1, Z_2\right]= \left\{ \beta (Z_1), \beta (Z_2)\right\}$, for   each $Z_1,Z_2 \in L$;
		
		\setlength{\itemsep}{4pt}
		\item $\left\{ \left[Y_1 ,Y_2\right], Y_3 \right\}= \alpha (Y_1) \vartriangleright \left\{Y_2,Y_3\right\}+\left\{Y_1,\left[Y_2,Y_3\right]\right\}-\alpha (Y_2)\vartriangleright \left\{Y_1,Y_3\right\}-\left\{Y_2,\left[Y_1,Y_3\right]\right\}$, for each $Y_1,Y_2, Y_3 \in \mathcal{h}$. This is the same as :
		$$\left\{\left[Y_1,Y_2\right],Y_3\right\}= \left\{\alpha (Y_1)\vartriangleright Y_2, Y_3\right\} - \left\{ \alpha (Y_2)\vartriangleright Y_1, Y_3\right\} - \left\{Y_1, \beta \left\{ Y_2, Y_3\right\}\right\} + \left\{Y_2, \beta\left\{Y_1,Y_2 \right\}\right\};$$
		
		\setlength{\itemsep}{4pt}
		\item $ \left\{Y_1,\left[Y_2,Y_3\right]\right\}= \left\{ \beta\left\{Y_1,Y_2 \right\},Y_3  \right\}-\left\{ \beta\left\{Y_1,Y_3 \right\},Y_2 \right\}$ for each $Y_1,Y_2, Y_3 \in \mathcal{h}$;
		
		\setlength{\itemsep}{4pt}
		\item $\left\{ \beta (Z) , Y \right\} +\left\{Y, \beta (Z)\right\} =- (\alpha (Y)\vartriangleright Z)$, for each $Y\in \mathcal{h} $ and $ Z \in \mathcal
		{l}$.	
	\end{enumerate}
	Analogously to the Lie 2-crossed module case, there is a left action of $\mathcal{h}$ on $\mathcal{l}$ which is defined by
	\begin{equation}\label{YZZ}
		Y\vartriangleright'Z= -\left\{ \beta(Z),Y\right\},
	\end{equation}
	for each $Y\in \mathcal{h}$ and $Z \in \mathcal{l}$. This together with the homomorphism $\beta :\mathcal{l} \longrightarrow \mathcal{h}$ defines a differential crossed module $(L, H; \beta, \vartriangleright')$.
	If $(\mathcal{h}, \mathcal{g}; \alpha, \vartriangleright)$ is also a differential crossed module and have
	\begin{align}
		\alpha (Y)\vartriangleright Z= Y\vartriangleright'Z,
	\end{align}
	in $(\mathcal{l},\mathcal{h}, \mathcal{g}; \beta, \alpha,\vartriangleright, \left\{ , \right\})$, we call this kind of differential 2-crossed module fine.
	
	\textbf{Mixed relations.}
	Let $(H, G; \alpha, \vartriangleright)$ be a Lie crossed module, and let $(\mathcal{h}, \mathcal{g}; \alpha, \vartriangleright)$ be
	the associated differential crossed module. 
	Therefore $G$ acts on $\mathcal{g}$ by the adjoint
	action $Ad: G \times \mathcal{g} \longrightarrow \mathcal{g}$, 
	$Ad_g X= g X g^{-1}$,
	and on $\mathcal{h}$ by the action $\vartriangleright: G \times \mathcal{h} \longrightarrow \mathcal{h}$, obeying the following algebraic identities:
	\begin{align}
		&	(g_1 g_2)\vartriangleright Y = g_1 \vartriangleright (g_2 \vartriangleright Y),\\
		& g \vartriangleright (X \vartriangleright Y)= (Ad_g X)\vartriangleright (g \vartriangleright Y) ,\\
		&\alpha(g\vartriangleright Y)= Ad_g \alpha(Y),\label{gY}\\
		&	\alpha (h) \vartriangleright Y = h Y h^{-1},
	\end{align}
	for any $g$, $g_1$, $g_2 \in G$, $X \in \mathcal{g}$, $Y \in \mathcal{h}$, and $h \in H $. 
	
	Let $(L, H, G;\beta, \alpha, \vartriangleright, \left\{,\right\})$ be a Lie 2-crossed module, and $(\mathcal{l},\mathcal{h},\mathcal{g};\beta,\alpha,\vartriangleright,\left\{,\right\})$ be the associated differential 2-crossed module. Except the above identities, there are another mixed relations given by the action of $G$ on $\mathcal{l}$, $\vartriangleright: G \times \mathcal{l} \longrightarrow \mathcal{l}$,
	\begin{align}
		&\beta(g \vartriangleright Z)=g \vartriangleright \beta(Z),\\
		&(g_1 g_2)\vartriangleright Z = g_1 \vartriangleright (g_2 \vartriangleright Z),\\
		& g \vartriangleright (X \vartriangleright Z)= (Ad_g X)\vartriangleright (g \vartriangleright Z) ,\\
		&g \vartriangleright\{Y_1, Y_2\}=\{g \vartriangleright Y_1, g \vartriangleright Y_2\}, \label{gyy}
	\end{align}
	for any $g$, $g_1$, $g_2 \in G$, $X \in \mathcal{g}$, $Z \in \mathcal{l}$.  
	

	\section{Gauge invariance of 2- and 3-Bianchi Identities}\label{2gt}
	\paragraph{\textbf{2-gauge invariance of 2-Bianchi Identities}}
	
	Under the general 2-gauge transformation \eqref{33},
	\begin{align}
		&d \Omega''_1 + A'' \wedge^{[, ]}\Omega''_1 + \alpha(\Omega''_2)     \nonumber\\
		&= d(g^{-1} \Omega_1g) + (g^{-1}A g + g^{-1}dg + \alpha(\phi)) \wedge (g^{-1} \Omega_1 g)- (g^{-1} \Omega_1 g)\wedge (g^{-1}A g + g^{-1}dg + \alpha(\phi))  \nonumber\\
		&\ \ \ + \alpha(g^{-1}\vartriangleright \Omega_2 + (g^{-1}\Omega_1 g)\wedge^{\vartriangleright}\phi )  \nonumber\\
		&= dg^{-1} \wedge \Omega_1 g + g^{-1}d \Omega_1 g + g^{-1} \Omega_1 \wedge dg + g^{-1} A \wedge \Omega_1 g + g^{-1}dg g^{-1}\wedge \Omega_1 g + \alpha(\phi) \wedge g^{-1} \Omega_1 g \nonumber\\
		&\ \ \ - g^{-1}\Omega_1 \wedge A g - g^{-1}\Omega_1 \wedge dg - (g^{-1}\Omega_1 g) \wedge \alpha(\phi) + \alpha(g^{-1}\vartriangleright \Omega_2) + (g^{-1}\Omega_1 g) \wedge^{[, ]}\alpha(\phi)\nonumber\\
		&= g^{-1}(d \Omega_1 + A \wedge^{[, ]} \Omega_1 + \alpha(\Omega_2))g \nonumber\\
		&=0,
	\end{align}
	and 
	\begin{align}
		&d \Omega''_2 + A'' \wedge^{\vartriangleright}\Omega''_2 - \Omega''_1 \wedge^{\vartriangleright}B''\nonumber\\
		&= d(g^{-1}\vartriangleright \Omega_2 + (g^{-1}\Omega_1g)\wedge^{\vartriangleright}\phi  ) + (g^{-1}A g + g^{-1}dg + \alpha(\phi))\wedge^{\vartriangleright}(g^{-1} \vartriangleright \Omega_2 + (g^{-1}\Omega_1g )\wedge^{\vartriangleright}\phi )\nonumber\\
		&\ \ \ -  (g^{-1}\Omega_1g )\wedge^{\vartriangleright} (g^{-1}\vartriangleright B + d \phi + (g^{-1}A g + g^{-1}dg + \alpha(\phi))\wedge^{\vartriangleright}\phi - \phi \wedge \phi)\nonumber\\
		&= dg^{-1} \vartriangleright \Omega_2 + g^{-1} \vartriangleright d \Omega_2 + (dg^{-1}\Omega_1 g + g^{-1}d \Omega_1 g + g^{-1} \Omega_1 dg )\wedge^{\vartriangleright}\phi  + (g^{-1}\Omega_1 g ) \wedge^{\vartriangleright}d \phi\nonumber\\
		&\ \ \ (g^{-1}A)\wedge^{\vartriangleright}\Omega_2 + (g^{-1}dg g^{-1})\wedge^{\vartriangleright}\Omega_2 + (\alpha(\phi)g^{-1})\wedge^{\vartriangleright}\Omega_2 + (g^{-1}A \wedge \Omega_1 g)\wedge^{\vartriangleright}\phi \nonumber\\
		&\ \ \ + (g^{-1} dg g^{-1}\wedge \Omega_1 g)\wedge^{\vartriangleright}\phi + (\alpha(\phi)\wedge (g^{-1}\Omega_1 g))\wedge^{\vartriangleright}\phi - (g^{-1}\Omega_1)\wedge^{\vartriangleright}B - (g^{-1} \Omega_1 g)\wedge^{\vartriangleright}d\phi \nonumber\\
		&\ \ \ - (g^{-1} \Omega_1 \wedge A g)\wedge^{\vartriangleright}\phi - (g^{-1}\Omega_1 \wedge dg)\wedge^{\vartriangleright}\phi - ((g^{-1}\Omega_1 g) \wedge \alpha(\phi))\wedge^{\vartriangleright} \phi + (g^{-1}\Omega_1 g)\wedge^{\vartriangleright}(\phi \wedge \phi)\nonumber\\
		&= g^{-1} \vartriangleright (d \Omega_2 + A \wedge^{\vartriangleright} \Omega_2 - \Omega_1 \wedge^{\vartriangleright} B) + g^{-1}(d \Omega_1 + A \wedge^{[, ]}\Omega_1)g \wedge^{\vartriangleright} \phi  + (\alpha(\phi)g^{-1})\wedge^{\vartriangleright} \Omega_2 \nonumber\\
		&\ \ \ + (\alpha(\phi) \wedge (g^{-1}\Omega_1 g))\wedge^{\vartriangleright}\phi - ((g^{-1}\Omega_1 g) \wedge \alpha(\phi))\wedge^{\vartriangleright}\phi + (g^{-1}\Omega_1 g)\wedge^{\vartriangleright}(\phi \wedge \phi)\nonumber\\
		&=0,
	\end{align}
	by using the 2-Bianchi-Identities \eqref{2BI}
	and the following identities:
	\begin{align}
		&	dg^{-1}= - g^{-1}dg g^{-1},\\
		& \alpha((g^{-1}\Omega_1 g)\wedge^{\vartriangleright}\phi) = (g^{-1}\Omega_1 g)\wedge^{[, ]}\alpha(\phi),\\
		&\alpha(g^{-1} \vartriangleright \Omega_2)= g^{-1}\alpha(\Omega_2)g,\\
		&\alpha(\phi)\wedge^{\vartriangleright}(g^{-1} \vartriangleright \Omega_2)=\alpha(g^{-1}\vartriangleright \Omega_2)\wedge^{\vartriangleright}\phi,\\
		& \alpha((g^{-1} \Omega_1 g)\wedge^{\vartriangleright}\phi)\wedge^{\vartriangleright}\phi = ((g^{-1}\Omega_1 g)\wedge^{\vartriangleright}\phi)^{[, ]}\phi,\\
		& (g^{-1}\Omega_1 g) \wedge^{\vartriangleright}(\phi \wedge \phi)=((g^{-1}\Omega_1 g) \wedge^{\vartriangleright} \phi)\wedge \phi + \phi \wedge ((g^{-1}\Omega_1 g) \wedge^{\vartriangleright}\phi),
	\end{align}
	which follow from \eqref{XxY}, \eqref{YyY'}, \eqref{gY}.
	
	\paragraph{\textbf{3-gauge invariance of 3-Bianchi Identities}}
	Under the general 3-gauge transformation \eqref{g3gt},
	\begin{align}
		&\ \ \ d \overline{\Omega}_1 + \overline{A} \wedge^{[, ]} \overline{\Omega}_1 + \alpha(\overline{\Omega}_2) \nonumber\\
		&= d(g^{-1} \Omega_1 g) + (g^{-1}Ag + g^{-1}dg + \alpha(\phi))\wedge^{[, ]}(g^{-1}\Omega_1 g) + \alpha(g^{-1}\vartriangleright \Omega_2 + \overline{\Omega}_1 \wedge^{\vartriangleright}\phi) \nonumber\\
		&= d g^{-1} \wedge \Omega_1 g + g^{-1}d\Omega_1 g + g^{-1} \Omega_1 dg + g^{-1}A \wedge \Omega_1 g + g^{-1}dg g^{-1}\wedge \Omega_1 g + \alpha(\phi)\wedge (g^{-1}\Omega_1 g )\nonumber\\
		& \ \ \  - g^{-1}\Omega_1 \wedge A g - g^{-1}\Omega_1 \wedge dg - (g^{-1}\Omega_1 g) \wedge \alpha(\phi)+ g^{-1}\alpha(\Omega_2)g + (g^{-1}\Omega_1 g)\wedge^{[, ]}\alpha(\phi)\nonumber\\
		&=g^{-1}(d \Omega_1 + A \wedge^{[, ]}\Omega_1 + \alpha(\Omega_2))g\nonumber\\
		&=0,
	\end{align}
	and
	\begin{align}
		&\ \ \ d \overline{\Omega}_2 + \overline{A} \wedge^{\vartriangleright}\overline{\Omega}_2 - \overline{\Omega}_1\wedge^{\vartriangleright}\overline{B} + \beta(\overline{\Omega}_3)\nonumber\\
		&=d (g^{-1} \vartriangleright \Omega_2 + (g^{-1}\Omega_1 g) \wedge^{\vartriangleright}\phi) + (g^{-1}Ag + g^{-1}dg + \alpha(\phi))\wedge^{\vartriangleright}(g^{-1} \vartriangleright \Omega_2 + (g^{-1}\Omega_1 g) \wedge^{\vartriangleright}\phi)\nonumber\\
		&\ \ \ -(g^{-1}\Omega_1 g)\wedge^{\vartriangleright} (g^{-1}\vartriangleright B + d \phi + (g^{-1}Ag + g^{-1}dg + \alpha(\phi))\wedge^{\vartriangleright} \phi - \phi \wedge\phi - \beta(\psi))\nonumber\\
		&\ \ \ + \beta(g^{-1}\vartriangleright \Omega_3 - (g^{-1}\vartriangleright \Omega_2 + (g^{-1}\Omega_1 g) \wedge^{\vartriangleright}\phi)\wedge^{\{, \}}\phi + \phi \wedge^{\{, \}}(g^{-1}\vartriangleright \Omega_2)- (g^{-1}\Omega_1 g) \wedge^{\vartriangleright}\psi)\nonumber\\
		&=dg^{-1}\wedge^{\vartriangleright}\Omega_2 + g^{-1}\vartriangleright d \Omega_2 + (dg^{-1}\Omega_1 g + g^{-1}d \Omega_1 g + g^{-1}\Omega_1 dg)\wedge^{\vartriangleright}\phi + g^{-1}\Omega_1g \wedge^{\vartriangleright}d \phi \nonumber\\
		&\ \ \ +g^{-1}\vartriangleright(A \wedge^{\vartriangleright}\Omega_2) + g^{-1}dg g^{-1}\wedge^{\vartriangleright}\Omega_2 + \alpha(\phi)\wedge^{\vartriangleright}(g^{-1}\vartriangleright \Omega_2) + (g^{-1}A \wedge \Omega_1 g)\wedge^{\vartriangleright}\phi \nonumber\\
		&\ \ \ + (g^{-1}dg g^{-1}\Omega_1 g)\wedge^{\vartriangleright}\phi + (\alpha(\phi)g^{-1}\Omega_1 g)\wedge^{\vartriangleright}\phi - g^{-1} \vartriangleright(\Omega_1\wedge^{\vartriangleright}B) - (g^{-1}\Omega_1 g) \wedge^{\vartriangleright}d \phi \nonumber\\
		&\ \ \ - (g^{-1}\Omega_1 \wedge A g)\wedge^{\vartriangleright}\phi - (g^{-1}\Omega_1dg)\wedge^{\vartriangleright}\phi - (g^{-1}\Omega_1g)\wedge^{\vartriangleright}(\alpha(\phi)\wedge^{\vartriangleright}\phi) + (g^{-1}\Omega_1g)\wedge^{\vartriangleright}(\phi \wedge \phi)\nonumber\\
		&\ \ \ + (g^{-1}\Omega_1g)\wedge^{\vartriangleright} \beta(\psi) + g^{-1}\vartriangleright \beta(\Omega_3) - \beta((g^{-1}\vartriangleright \Omega_2 + (g^{-1}\Omega_1 g) \wedge^{\vartriangleright}\phi)\wedge^{\{, \}}\phi) \nonumber\\
		&\ \ \ + \beta (\phi\wedge^{\{, \}}(g^{-1}\vartriangleright \Omega_2)) - (g^{-1}\Omega_1 g) \wedge^{\vartriangleright}\beta(\psi)\nonumber\\
		&=g^{-1}\vartriangleright (d \Omega_2 + A \wedge^{\vartriangleright}\Omega_2 - \Omega_1 \wedge^{\vartriangleright}B + \beta(\Omega_3))\nonumber\\
		&=0,
	\end{align}
	and 
	\begin{align}
		&\ \ \ d \overline{\Omega}_3 + \overline{A} \wedge^{\vartriangleright}\overline{\Omega}_3 - \overline{\Omega}_1\wedge^{\vartriangleright}\overline{C} - \overline{B}\wedge^{\{, \}}\overline{\Omega}_2 - \overline{\Omega}_2 \wedge^{\{, \}}\overline{B}\nonumber\\
		&= d(g^{-1}\vartriangleright \Omega_3 - (g^{-1}\vartriangleright \Omega_2 + (g^{-1}\Omega_1 g) \wedge^{\vartriangleright}\phi)\wedge^{\{, \}}\phi + \phi \wedge^{\{, \}}(g^{-1}\vartriangleright \Omega_2)- (g^{-1}\Omega_1 g) \wedge^{\vartriangleright}\psi)\nonumber\\
		&\ \ \ +(g^{-1}A g + g^{-1}dg + \alpha(\phi))\wedge^{\vartriangleright} (g^{-1}\vartriangleright \Omega_3 - (g^{-1}\vartriangleright \Omega_2 + (g^{-1}\Omega_1 g) \wedge^{\vartriangleright}\phi)\wedge^{\{, \}}\phi \nonumber\\
		&\ \ \ + \phi \wedge^{\{, \}}(g^{-1}\vartriangleright \Omega_2)- (g^{-1}\Omega_1 g) \wedge^{\vartriangleright}\psi)- (g^{-1}\Omega_1g )\wedge^{\vartriangleright}(g^{-1}\vartriangleright C - (g^{-1}\vartriangleright B + d \phi + \overline{A}\wedge^{\vartriangleright}\phi \nonumber\\
		&\ \ \ - \phi \wedge \phi -\beta(\psi))\wedge^{\{, \}} \phi + \phi \wedge^{\vartriangleright'}\psi - \phi \wedge^{\{, \}}(g^{-1}\vartriangleright B)- d \psi - (g^{-1}A g + g^{-1}dg + \alpha(\phi)) \wedge^{\vartriangleright} \psi)\nonumber\\
		&\ \ \ - (g^{-1}\vartriangleright B + d \phi +  (g^{-1}A g + g^{-1}dg + \alpha(\phi))\wedge^{\vartriangleright}\phi - \phi \wedge \phi -\beta(\psi)) \wedge^{\{, \}}(g^{-1}\vartriangleright \Omega_2 + \overline{\Omega}_1\wedge^{\vartriangleright}\phi)\nonumber\\
		&\ \ \ - (g^{-1}\vartriangleright \Omega_2 + \overline{\Omega}_1\wedge^{\vartriangleright}\phi)\wedge^{\{, \}}(g^{-1}\vartriangleright B + d \phi +  (g^{-1}A g + g^{-1}dg + \alpha(\phi))\wedge^{\vartriangleright}\phi - \phi \wedge \phi -\beta(\psi))\nonumber\\
		&= g^{-1}\vartriangleright(d \Omega_3 + A \wedge^{\vartriangleright}\Omega_3 - \Omega_1\wedge^{\vartriangleright}C - B\wedge^{\{, \}}\Omega_2 - \Omega_2 \wedge^{\{, \}}B )\nonumber\\
		&= 0,
	\end{align}
	by using the definition of differential 2-crossed module \eqref{2cm}, the 3-Bianchi-Identities \eqref{3bi}
	and the following identities:
	\begin{align}
		&	dg^{-1}= - g^{-1}dg g^{-1},\\
		&\alpha(\phi)\wedge^{\vartriangleright}(g^{-1}\vartriangleright \Omega_2)= \phi\wedge^{[, ]}(g^{-1}\vartriangleright \Omega_2)= \alpha(g^{-1}\vartriangleright \Omega_2)\wedge^{\vartriangleright}\phi= (g^{-1}\vartriangleright \alpha(\Omega_2))\wedge^{\vartriangleright}\phi, \\
		&(g^{-1}\Omega_1g)\wedge^{\vartriangleright}(\alpha(\phi)\wedge^{\vartriangleright}\phi)= (g^{-1}\Omega_1g)\wedge^{\vartriangleright}(\phi \wedge^{[, ]}\phi)= 2 (g^{-1}\Omega_1g)\wedge^{\vartriangleright}(\phi \wedge \phi),\\
		&(\alpha(\phi)\wedge (g^{-1}\Omega_1g))\wedge^{\vartriangleright}\phi = \alpha(\phi)\wedge^{\vartriangleright}((g^{-1}\Omega_1g )\wedge^{\vartriangleright}\phi)=\phi \wedge^{[, ]}((g^{-1}\Omega_1g) \wedge^{\vartriangleright}\phi),\\
		&(g^{-1}\Omega_1 g )\wedge^{\vartriangleright}(\phi\wedge \phi)= (g^{-1}\Omega_1g)\wedge^{\vartriangleright}\phi \wedge \phi + \phi\wedge(g^{-1} \Omega_1 g \wedge^{\vartriangleright} \phi)\\
		&\beta(\phi \wedge^{\{, \}}(g^{-1}\vartriangleright\Omega_2))= \phi \wedge^{[,]}(g^{-1}\vartriangleright\Omega_2)- \alpha(\phi) \wedge^{\vartriangleright}(g^{-1}\vartriangleright\Omega_2),\\
		& \beta((g^{-1}\vartriangleright \Omega_2)\wedge^{\{, \}}\phi)= (g^{-1}\vartriangleright \Omega_2)\wedge^{[, ]}\phi - \alpha(g^{-1}\vartriangleright \Omega_2)\wedge^{\vartriangleright}\phi ,\\
		&\beta(((g^{-1}\Omega_1g) \wedge^{\vartriangleright}\phi)\wedge^{\{, \}}\phi)= ((g^{-1}\Omega_1g) \wedge^{\vartriangleright}\phi)\wedge^{[, ]}\phi- \alpha((g^{-1}\Omega_1g) \wedge^{\vartriangleright}\phi)\wedge^{\vartriangleright}\phi,
	\end{align}
	which follow from \eqref{YyY'}, \eqref{X_1X_2Y}, \eqref{XY_1Y_2}, \eqref{gyy}, \eqref{12}.
	
	\section{The another proofs of the 2- and 3-Chern-Weil theorem}\label{23cw}
	\paragraph{\textbf{The another proof of the 2-Chern-Weil theorem}}
	\begin{proof}
		\begin{align}
			&2\langle \Omega^1_1, \Omega^1_2 \rangle_{\mathcal{g}, \mathcal{h}} - 2\langle \Omega^0_1, \Omega^0_2\rangle_{\mathcal{g}, \mathcal{h}}\nonumber\\
			&= 2(\langle d A^1 + A^1 \wedge A^1 - \alpha(B^1),  dB^1 + A^1 \wedge^{\vartriangleright}B^1 \rangle_{\mathcal{g}, \mathcal{h}}\nonumber\\
			&\ \  \ - \langle d A^0 + A^0\wedge A^0 - \alpha(B^0), dB^0 + A^0 \wedge^{\vartriangleright}B^0 \rangle_{\mathcal{g}, \mathcal{h}})\nonumber\\
			&= 2(\langle d A^1, dB^1 \rangle_{\mathcal{g}, \mathcal{h}} + \langle A^1 \wedge A^1, d B^1 \rangle_{\mathcal{g}, \mathcal{h}} - \langle \alpha(B^1), d B^1 \rangle_{\mathcal{g}, \mathcal{h}} + \langle d A^1, A^1 \wedge^{\vartriangleright} B^1 \rangle_{\mathcal{g}, \mathcal{h}}\nonumber\\
			&\ \ \ -\langle d A^0, dB^0 \rangle_{\mathcal{g}, \mathcal{h}} - \langle A^0 \wedge A^0, d B^0 \rangle_{\mathcal{g}, \mathcal{h}} + \langle \alpha(B^0), d B^0 \rangle_{\mathcal{g}, \mathcal{h}} - \langle d A^0, A^0 \wedge^{\vartriangleright} B^0 \rangle_{\mathcal{g}, \mathcal{h}}),
		\end{align}
		by using \eqref{XXY}.  Consider  \eqref{AB} and \eqref{Q4}, have
		\begin{align}
			&d Q(A^1, A^0, B^1, B^0)\nonumber\\
			&=2 \int_{0}^{1}dt (\langle d A^1 - dA^0, \Omega^t_2 \rangle_{\mathcal{g}, \mathcal{h}} - \langle A^1 - A^0, d \Omega^t_2 \rangle_{\mathcal{g}, \mathcal{h}}\nonumber\\
			&\ \ \ + \langle d \Omega^t_1, B^1 - B^0 \rangle_{\mathcal{g}, \mathcal{h}} + \langle \Omega^t_1, d B^1 - dB^0 \rangle_{\mathcal{g}, \mathcal{h}})\nonumber\\
			&= 2 ( \langle dA^1 - dA^0, \frac{1}{2}dB^0 + \frac{1}{2}dB^1 + \frac{1}{3} A^0 \wedge^{\vartriangleright} B^0 + \frac{1}{6} A^0 \wedge^{\vartriangleright}B^1 + \frac{1}{6} A^1 \wedge^{\vartriangleright}B^0 + \frac{1}{3} A^1 \wedge^{\vartriangleright}B^1 \rangle_{\mathcal{g}, \mathcal{h}}\nonumber\\
			&\ \ \ - \langle A^1 -A^0, \frac{1}{3} d A^0 \wedge^{\vartriangleright} B^0 + \frac{1}{6}dA^0 \wedge^{\vartriangleright}B^1+ \frac{1}{6}dA^1 \wedge^{\vartriangleright}B^0 + \frac{1}{3} d A^1 \wedge^{\vartriangleright} B^1 - \frac{1}{3}  A^0 \wedge^{\vartriangleright} d B^0 \nonumber\\
			&\ \ \ - \frac{1}{6}A^0 \wedge^{\vartriangleright}dB^1 - \frac{1}{6}A^1 \wedge^{\vartriangleright}dB^0 - \frac{1}{3}A^1 \wedge^{\vartriangleright}dB^1 \rangle_{\mathcal{g}, \mathcal{h}}\nonumber\\
			&\ \ \ + \langle \frac{1}{3}dA^0 \wedge A^0 + \frac{1}{6}dA^0 \wedge A^1 +  \frac{1}{6}dA^1 \wedge A^0 +  \frac{1}{3}dA^1 \wedge A^1 -  \frac{1}{3}A^0 \wedge dA^0 - \frac{1}{6}A^0 \wedge dA^1 -  \frac{1}{6}A^1 \wedge dA^0 \nonumber\\
			&\ \ \ -  \frac{1}{3}A^1 \wedge dA^1 - \frac{1}{2} \alpha(d B^0) - \dfrac{1}{2} \alpha(d B^1), B^1 - B^0 \rangle_{\mathcal{g}, \mathcal{h}} \nonumber\\
			&\ \ \  + \langle \frac{1}{2} dA^0 + \frac{1}{2}d A^1 + \frac{1}{3}A^0 \wedge A^0 + \frac{1}{6}A^0 \wedge A^1 + \frac{1}{6}A^1 \wedge A^0 + \frac{1}{3}A^1 \wedge A^1 -\frac{1}{2}\alpha(B^0) - \frac{1}{2}\alpha(B^1),\nonumber\\
			&\ \ \  d B^1 - dB^0 \rangle_{\mathcal{g}, \mathcal{h}})\nonumber\\
			&=2(\langle d A^1, dB^1 \rangle_{\mathcal{g}, \mathcal{h}} + \langle A^1 \wedge A^1, d B^1 \rangle_{\mathcal{g}, \mathcal{h}} - \langle \alpha(B^1), d B^1 \rangle_{\mathcal{g}, \mathcal{h}} + \langle d A^1, A^1 \wedge^{\vartriangleright} B^1 \rangle_{\mathcal{g}, \mathcal{h}}\nonumber\\
			&\ \ \ -\langle d A^0, dB^0 \rangle_{\mathcal{g}, \mathcal{h}} - \langle A^0 \wedge A^0, d B^0 \rangle_{\mathcal{g}, \mathcal{h}} + \langle \alpha(B^0), d B^0 \rangle_{\mathcal{g}, \mathcal{h}} - \langle d A^0, A^0 \wedge^{\vartriangleright} B^0 \rangle_{\mathcal{g}, \mathcal{h}}),
		\end{align}
		by using \eqref{XXY}. Then, we have
		\begin{align}
			2\langle \Omega^1_1, \Omega^1_2 \rangle_{\mathcal{g}, \mathcal{h}} - 2\langle \Omega^0_1, \Omega^0_2\rangle_{\mathcal{g}, \mathcal{h}}=d Q(A^1, A^0, B^1, B^0)
		\end{align}
	\end{proof}
	
	\paragraph{\textbf{The another proof of the 3-Chern-Weil theorem}}
	\begin{proof}
		We have
		\begin{align}\label{dQ5}
			&d Q(A^1, A^0, B^1, B^0, C^1, C^0)\nonumber\\
			&=2 \int_{0}^{1}dt( \langle dA^1- dA^0, \Omega^t_3 \rangle_{\mathcal{g}, \mathcal{l}}- \langle A^1- A^0, d\Omega^t_3 \rangle_{\mathcal{g}, \mathcal{l}} + \langle 	 dB^1 - dB^0, \Omega^t_2 \rangle_{\mathcal{h}} \nonumber\\
			&\ \ \ + \langle 	 B^1 - B^0, d\Omega^t_2 \rangle_{\mathcal{h}}+ \langle d\Omega^t_1, C^1 - C^0 \rangle_{\mathcal{g}, \mathcal{l}} +  \langle \Omega^t_1, dC^1 - dC^0 \rangle_{\mathcal{g}, \mathcal{l}})\nonumber\\
			&= \langle dA^1- dA^0, 2 \int_{0}^{1}dt \Omega^t_3 \rangle_{\mathcal{g}, \mathcal{l}}- \langle A^1- A^0, 2 \int_{0}^{1}dt d\Omega^t_3 \rangle_{\mathcal{g}, \mathcal{l}} + \langle 	 dB^1 - dB^0, 2 \int_{0}^{1}dt \Omega^t_2 \rangle_{\mathcal{h}} \nonumber\\
			&\ \ \ + \langle 	 B^1 - B^0, 2 \int_{0}^{1}dt d\Omega^t_2 \rangle_{\mathcal{h}}+ \langle 2 \int_{0}^{1}dt d\Omega^t_1, C^1 - C^0 \rangle_{\mathcal{g}, \mathcal{l}} +  \langle 2 \int_{0}^{1}dt \Omega^t_1, dC^1 - dC^0 \rangle_{\mathcal{g}, \mathcal{l}}.
		\end{align}
		We can calculate 
		\begin{align}
			2 \int_{0}^{1}dt \Omega^t_3= &dC^0 + dC^1 + \frac{2}{3}A^0 \wedge C^0 + \frac{1}{3}A^0 \wedge^{\vartriangleright}C^1 + \frac{1}{3}A^1 \wedge^{\vartriangleright}C^0 + \frac{2}{3}A^1 \wedge^{\vartriangleright}C^1 \nonumber\\
			& + \frac{2}{3}B^0 \wedge^{\{, \}} B^0 + \frac{1}{3}B^0 \wedge^{\{, \}} B^1  + \frac{1}{3}B^1 \wedge^{\{, \}} B^0  + \frac{2}{3}B^1 \wedge^{\{, \}} B^1,\label{j3}\\
			2 \int_{0}^{1}dt d\Omega^t_3= &\frac{2}{3} dA^0 \wedge^{\vartriangleright}C^0 + \frac{1}{3} dA^0 \wedge^{\vartriangleright}C^1 + \frac{1}{3} dA^1 \wedge^{\vartriangleright}C^0 + \frac{2}{3} dA^1 \wedge^{\vartriangleright}C^1 \nonumber\\
			& - \frac{2}{3} A^0 \wedge^{\vartriangleright}d C^0 - \frac{1}{3} A^0 \wedge^{\vartriangleright}d C^1 - \frac{1}{3} A^1 \wedge^{\vartriangleright}d C^0 - \frac{2}{3} A^1 \wedge^{\vartriangleright}d C^1 \nonumber\\
			& + \frac{2}{3} d B^0 \wedge^{\{, \}}B^0 + \frac{1}{3} d B^0 \wedge^{\{, \}}B^1 + \frac{1}{3} d B^1 \wedge^{\{, \}}B^0 + \frac{2}{3} d B^1 \wedge^{\{, \}}B^1 \nonumber\\
			&+ \frac{2}{3} B^0 \wedge^{\{, \}} dB^0 + \frac{1}{3} B^0 \wedge^{\{, \}}dB^1 + \frac{1}{3} B^1 \wedge^{\{, \}}dB^0 
			+ \frac{2}{3} B^1 \wedge^{\{, \}}dB^1, \label{dj3}
		\end{align}
		\begin{align}
			2 \int_{0}^{1}dt \Omega^t_2= & dB^0 +dB^1 + \frac{2}{3}A^0 \wedge^{\vartriangleright}B^0 + \frac{1}{3}A^0 \wedge^{\vartriangleright}B^1 + \frac{1}{3}A^1 \wedge^{\vartriangleright}B^0  \nonumber\\
			&+ \frac{2}{3}A^1 \wedge^{\vartriangleright}B^1 - \beta(C^0)- \beta(C^1), \label{j2}\\
			2 \int_{0}^{1}dt d\Omega^t_2= & \frac{2}{3}dA^0 \wedge^{\vartriangleright}B^0 + \frac{1}{3}dA^0 \wedge^{\vartriangleright}B^1 + \frac{1}{3}dA^1 \wedge^{\vartriangleright}B^0 + \frac{2}{3}dA^1 \wedge^{\vartriangleright}B^1- \frac{2}{3}A^0 \wedge^{\vartriangleright}d B^0  \nonumber\\
			&- \frac{1}{3}A^0 \wedge^{\vartriangleright}d B^1 - \frac{1}{3}A^1 \wedge^{\vartriangleright}d B^0 - \frac{2}{3}A^1 \wedge^{\vartriangleright}d B^1 - \beta(dC^0)- \beta(dC^1),\label{dj2}
		\end{align}
		and
		\begin{align}
			2 \int_{0}^{1}dt \Omega^t_1= & dA^0 + dA^1 + \frac{2}{3}A^0 \wedge A^0 + \frac{1}{3} A^0 \wedge A^1 + \frac{1}{3} A^1 \wedge A^0 + \frac{2}{3}A^1 \wedge A^1 - \alpha(B^0)- \alpha(B^1),\label{j1}\\
			2 \int_{0}^{1}dt d\Omega^t_1= & \frac{2}{3}d A^0 \wedge A^0 + \frac{1}{3} dA^0 \wedge A^1 + \frac{1}{3} dA^1 \wedge A^0 + \frac{2}{3}dA^1 \wedge A^1 	-\frac{2}{3} A^0 \wedge dA^0 -\nonumber\\
			& \frac{1}{3} A^0 \wedge dA^1 - \frac{1}{3} A^1 \wedge dA^0 - \frac{2}{3}A^1 \wedge dA^1 
			- \alpha(dB^0)- \alpha(dB^1).\label{dj1}
		\end{align}
		Substitute \eqref{j3}, \eqref{dj3}, \eqref{j2}, \eqref{dj2}, \eqref{j1} and \eqref{dj1} to \eqref{dQ5} to get
		\begin{align}
			&d Q_5(A^1, A^0, B^1, B^0, C^1, C^0)\nonumber\\
			&=2 \langle d A^1 + A^1\wedge A^1 - \alpha(B^1), d C^1 + A^1 \wedge^{\vartriangleright}C^1 + B^1 \wedge^{\{, \}}B^1 \rangle_{\mathcal{g}, \mathcal{l}}\nonumber\\
			&+ \langle dB^1 + A^1 \wedge^{\vartriangleright} B^1 - \beta(C^1), dB^1 + A^1 \wedge^{\vartriangleright} B^1 - \beta(C^1)\rangle_{\mathcal{h}}\nonumber\\
			&- 2 \langle d A^0 + A^0\wedge A^0- \alpha(B^0), d C^0 + A^0 \wedge^{\vartriangleright}C^0+ B^0 \wedge^{\{, \}}B^0 \rangle_{\mathcal{g}, \mathcal{l}} \nonumber\\
			&- \langle dB^0 + A^0 \wedge^{\vartriangleright} B^0 - \beta(C^0), dB^0 + A^0 \wedge^{\vartriangleright} B^0 - \beta(C^0)\rangle_{\mathcal{h}}\nonumber\\
			&= 	2 \langle \Omega^1_1, \Omega^1_3 \rangle_{\mathcal{g}, \mathcal{l}} + \langle \Omega^1_2, \Omega^1_2 \rangle_{\mathcal{h}}
			-
			2 \langle \Omega^0_1, \Omega^0_3 \rangle_{\mathcal{g}, \mathcal{l}} - \langle \Omega^0_2, \Omega^0_2 \rangle_{\mathcal{h}}.
		\end{align}
	\end{proof}

\end{document}